\documentclass[english,notitlepage,superscriptaddress,nofootinbib]{revtex4-1}
\usepackage[T1]{fontenc}
\usepackage[latin9]{inputenc}
\usepackage{geometry}
\usepackage{color}
\usepackage{babel}
\usepackage{mathtools}
\usepackage{amsmath}
\usepackage{amsthm}
\usepackage{amssymb}
\usepackage{graphicx}
\usepackage[unicode=true,pdfusetitle,
 bookmarks=true,bookmarksnumbered=false,bookmarksopen=false,
 breaklinks=false,pdfborder={0 0 0},pdfborderstyle={},backref=false,colorlinks=true]
 {hyperref}
\hypersetup{
 urlcolor=burntgreen,linkcolor=coolblack,citecolor=red}

\makeatletter

\providecommand{\tabularnewline}{\\}

\theoremstyle{plain}
    \ifx\thechapter\undefined
      \newtheorem{lem}{\protect\lemmaname}
    \else
      \newtheorem{lem}{\protect\lemmaname}[chapter]
    \fi
\theoremstyle{plain}
    \ifx\thechapter\undefined
      \newtheorem{prop}{\protect\propositionname}
    \else
      \newtheorem{prop}{\protect\propositionname}[chapter]
    \fi

\usepackage{xcolor}
\usepackage{newtxtext}
\usepackage{newtxmath}
\usepackage{dsfont}
\usepackage{bbm}

\definecolor{burntorange}{rgb}{0.8, 0.33, 0.0}
\definecolor{charcoal}{rgb}{0.21, 0.27, 0.31}
\definecolor{coolblack}{rgb}{0.0, 0.28, 0.49}
\definecolor{burntgreen}{rgb}{0.05, 0.45, 0.27}
\definecolor{burntblue}{rgb}{0.05, 0.27, 0.8}

\DeclareMathOperator*{\argmax}{argmax}

\newtheorem*{ub}{Upper Bound}
\newtheorem*{lb}{Lower Bound}

\makeatother

\providecommand{\lemmaname}{Lemma}
\providecommand{\propositionname}{Proposition}

\begin{document}
\title{Clustering of solutions in the symmetric binary perceptron}
\author{Carlo Baldassi}
\affiliation{Artificial Intelligence Lab, Institute for Data Science and Analytics,
Bocconi University, Milano, Italy}
\affiliation{Istituto Nazionale di Fisica Nucleare, Sezione di Torino, Italy}
\author{Riccardo Della Vecchia}
\affiliation{Artificial Intelligence Lab, Institute for Data Science and Analytics,
Bocconi University, Milano, Italy}
\author{Carlo Lucibello}
\affiliation{Artificial Intelligence Lab, Institute for Data Science and Analytics,
Bocconi University, Milano, Italy}
\author{Riccardo Zecchina}
\affiliation{Artificial Intelligence Lab, Institute for Data Science and Analytics,
Bocconi University, Milano, Italy}
\affiliation{International Centre for Theoretical Physics, Trieste, Italy}
\begin{abstract}
The geometrical features of the (non-convex) loss landscape of neural
network models are crucial in ensuring successful optimization and,
most importantly, the capability to generalize well. While minimizers'
flatness consistently correlates with good generalization, there has
been little rigorous work in exploring the condition of existence
of such minimizers, even in toy models. Here we consider a simple
neural network model, the symmetric perceptron, with binary weights.
Phrasing the learning problem as a constraint satisfaction problem,
the analogous of a flat minimizer becomes a large and dense cluster
of solutions, while the narrowest minimizers are isolated solutions.
We perform the first steps toward the rigorous proof of the existence
of a dense cluster in certain regimes of the parameters, by computing
the first and second moment upper bounds for the existence of pairs
of arbitrarily close solutions. Moreover, we present a non rigorous
derivation of the same bounds for sets of $y$ solutions at fixed
pairwise distances.
\end{abstract}
\maketitle

\section{Introduction}

The problem of learning to classify a set random patterns with a \emph{binary
perceptron} has been a recurrent topic since the very beginning of
the statistical physics studies of neural networks models \citep{gardner1988optimal}.
The learning problem consists in finding the optimal binary assignments
of the connection weights which minimize the number of misclassifications
of the patterns. We shall refer to such set of optimal assignments
as the space of solutions of the perceptron. In spite of the extremely
simple architecture of the model, the learning task is highly non
convex and its geometrical features are believed to play a role also
in more complex neural architectures \citep{watkin1993statistical,seung1992statistical,engel2001statistical}.

For the case of random i.i.d. patterns, the space of solutions of
the binary perceptron is known to be dominated by an exponential number
of isolated solutions \citep{krauth1989storage} which lie at a large
mutual Hamming distances \citep{huang2013entropy,huang2014origin}
(golf course landscape). An even larger number of local minima have
been shown to exist \citep{horner1992dynamics}.

The study of how the number of these isolated solutions decreases
as more patterns are learned provides the correct prediction for the
so-called capacity of the binary perceptron, i.e. the maximum number
of random patterns that can be correctly classified. However, the
same analysis does not provide the insight necessary for understanding
the behavior of learning algorithms: one would expect that finding
solutions in a golf course landscape should be difficult for search
algorithms, and indeed Monte Carlo based algorithms satisfying detailed
balance get stuck in local minima; yet, empirical results have shown
that many learning algorithms, even simple ones, are able to find
solutions efficiently \citep{braunstein2006learning,baldassi2007efficient,baldassi2009generalization,baldassi2015max}.

These empirical results suggested that the solutions which were not
the dominant ones in the Gibbs measure, and were as such neglected
in the analysis of the capacity, could in fact play an important algorithmic
role. As discussed in refs.~\citep{baldassi2015subdominant,baldassi2016local}
this turned out to be the actual case: the study of the dominant solutions
in the Gibbs measure theory does not take into account the existence
of rare (sub-dominant) regions in the solution space which are those
found by algorithms. Revealing those rare, accessible regions required
a large deviation analysis based on the notion of\emph{ local entropy},
which is a measure of the density of solutions in an extensive region
of the configuration space (see the precise definition in the next
section). The regions of maximal local entropy are extremely dense
in solutions, such that (for finite $N$) nearly every configuration
in the region is a solution. More recently, the existence of high
local entropy / flat regions has been found also in multi-layer networks
with continuous weights, and their role has been connected to the
structural characteristics of deep neural networks \citep{baldassi2019properties,baldassi2020shaping}.

All the above results rely on methods of statistical mechanics of
disordered systems which are extremely powerful and yet not fully
rigorous. It is therefore important to corroborate them with rigorous
bounds \citep{ding2019capacity}. In a recent paper \citep{aubin2019storage},
Aubin et al. have studied a simple variant of the binary perceptron
model for which the rigorous bounds provided by first and second moment
methods can be shown to be tight. The authors have been able to confirm
the predictions of the statistical physics methods concerning the
capacity of the model, and the golf course nature of the space of
solutions. The model that the authors have studied has a modified
activation criterion compared to the traditional perceptron, replacing
the Heaviside step function by a function with an even symmetry.

The goal of the present paper is to study the existence of dense regions
in the the symmetrized binary perceptron model. In sec.~\ref{sec:large-dev}
we define the model and, as a preliminary step, we present the results
of the replica-method large deviation analysis, which predicts that
the phenomenology for the symmetrized model is the same as for the
traditional one, and thus that high local entropy regions exist. If
these predictions are correct, then it should be possible, at least
for some range of the parameters, to choose any integer number $y\ge2$
and find a threshold $x_{c}\left(y\right)$ such that for any $x<x_{c}\left(y\right)$
there is an exponential number of groups of $y$ solutions all at
mutual Hamming distance $\left\lfloor Nx\right\rfloor $. In the remainder
of the paper we try to verify this statement, by employing the first
and second moment methods where possible. In sec.~\ref{sec:y=00003D2}
we address the $y=2$ case: we extend the analysis of ref.~\citep{aubin2019storage}
and show rigorously (except for a numerical optimization step) that,
for small enough constraint density $\alpha,$ there exist an exponential
number of pairs of solutions at arbitrary $O\left(N\right)$ Hamming
distance. In sec.~\ref{sec:Caseygen} we study the general $y$ case.
For $y=3$ or $4$, we can derive a rigorous upper bound that coincides
with the non-rigorous results for general $y$. As for the lower bound,
only the $y=2$ case can be derived rigorously (and again it coincides
with the non-rigorous results that we also derive). All the results
are thus consistent with the existence of high local entropy regions,
as predicted by the large deviation study.

\section{Dense clusters in the symmetric binary perceptron\label{sec:large-dev}}

\subsection{Model definition}

We investigate the rectangular-binary-perceptron (RBP) problem introduced
in ref. ~\citep{aubin2019storage}. The RBP has the key property
of having a symmetric activation function, characterized by a parameter
$K>0$. Given a vector of binary weights $\mathbf{w}\in\left\{ \pm1\right\} ^{N}$
and an input $\boldsymbol{\boldsymbol{\xi}}\in\mathbb{R}^{N}$ (an
example), we say that $\mathbf{w}$ satisfies the example if $\left|\boldsymbol{\xi}\cdot\mathbf{w}\right|<K$.\footnote{This setting corresponds to a binary classification problem with training
examples from a single class. This simplifies the analysis.} This symmetry simplifies the theoretical analysis and allows to obtain
tighter bounds for the storage capacity through the first and second
moment methods. 

For a given set of inputs $\boldsymbol{\boldsymbol{\xi}}^{\mu}\in\mathbb{R}^{N}$,
with $\mu=1,\dots,M$, the RBP problem can be expressed as a constraint
satisfaction problem (CSP) over the binary weights. Throughout the
paper we will assume the entries $\xi_{i}^{\mu}$ to be \emph{i.i.d.}
Gaussian variables with zero mean and variance $1/N$. A binary vector
$\mathbf{w}\in\left\{ \pm1\right\} ^{N}$ is called a \emph{solution
}of the problem if it satisfies

\begin{equation}
\sum_{i=1}^{N}w_{i}\xi_{i}^{\mu}\in I_{K}\qquad\forall\mu\in\left[M\right],\label{eq:recteq}
\end{equation}
where $I_{K}=\left[-K,K\right]$. 
\global\long\def\ind{\mathbbm{1}}%
 Equivalently, a vector $\mathbf{w}$ is a solution of the RBP problem
iff the function $\mathbb{X}_{\boldsymbol{\xi},K}:\left\{ -1,1\right\} ^{N}\to\left\{ 0,1\right\} $,
defined as

\begin{equation}
\mathbb{X}_{\boldsymbol{\xi},K}\left(\mathbf{w}\right)=\prod_{\mu=1}^{M}\ind\left(\,\sum_{i=1}^{N}w_{i}\xi_{i}^{\mu}\in I_{K}\,\right),\label{eq:CSP}
\end{equation}
is equal to one, where we have denoted with $\ind\left(p\right)$
an indicator function that is $1$ if the statement $p$ is true and
$0$ otherwise.

The \textit{storage capacity} is then defined similarly to the satisfiability
threshold in random constraint satisfaction problems: we denote the
constraint density as $\alpha\equiv M/N$ and define the storage capacity
$\alpha_{c}\left(K\right)$, also known as SAT-UNSAT transition point,
as the infimum of densities $\alpha$ such that, in the limit $N\to\infty$,
with high probability (over the choice of the matrix $\xi_{i}^{\mu}$)
there are no solutions. It is natural to conjecture that the converse
also holds, i.e. that the storage capacity $\alpha_{c}\left(K\right)$
equals the supremum of $\alpha$ such that in the limit $N\to\infty$
solutions exist with high probability. In this case we would say the
storage capacity is a \textit{sharp threshold}.

\subsection{Replicated Systems and Dense Clusters}

In order to obtain a geometric characterization of the solution space,
we consider the Hamming distance of any two configurations $\boldsymbol{w}^{1}$
and $\boldsymbol{w}^{2}$, defined by

\[
d_{\mathrm{H}}\left(\mathbf{w}^{1},\mathbf{w}^{2}\right)\equiv\sum_{i=1}^{N}\left(1-w_{i}^{1}w_{i}^{2}\right)/2.
\]

Even if an exponential number of solutions exist for $\alpha<\alpha_{c}\left(K\right)$,
the overwhelming majority are \emph{isolated}: for each such solution,
there exists a radius $r_{\min}$ such that the number of other solutions
within a distance $\left\lfloor Nr_{\min}\right\rfloor $ is sub-exponential.
We are interested instead in the presence of \emph{dense regions},
which are characterized by the fact that there is a configuration
around which the number of solutions within a given radius $\left\lfloor Nr\right\rfloor $
is exponential for all $r$ in some neighborhood of $0$. We speak
of \emph{ultra-dense} regions when the logarithm of the density of
solutions tends exponentially fast to $0$ as $r\to0$.

Suppose now that a dense region around some reference configuration
exists, choose a sufficiently small value $r>0$, and call $x$ the
typical distance between any two solutions at distance $r$ from the
reference. In general, $0<x\le2r$, and for an ultra-dense region
$x=2r\left(1-r\right)$ in the limit of large $N$. Therefore for
any $x$ below some threshold there should exist an exponential number
of solutions at mutual normalized distance $x$.

We thus investigate the problem of finding a set of $y$ solutions
of the RBP problem, where $y$ is an arbitrary natural number, with
all pairwise distances constrained to some value $\left\lfloor Nx\right\rfloor $.
The existence (for some range of $\alpha$) of such set of solutions,
w.h.p. in the large $N$ limit, for arbitrarily large values of $y$
and all $x$ in some neighborhood of $0$, is a necessary condition
for the presence of dense regions. These sets of $y$ solutions would
coexist with an exponentially larger number of isolated solutions,
and therefore the usual tools of statistical physics are not sufficient
to reveal their presence, and a large deviation analysis is necessary
\citep{baldassi2015subdominant}.

As a starting point for the analysis we introduce the partition function
of the model with $y$ real replicas, $\mathcal{Z}_{y}$, accounting
for the number of such sets (up to a $y!$ symmetry factor). For any
fixed (normalized) distance $x\in\left[0,1\right]$, this is given
by 

\begin{align}
\mathcal{Z}_{y}\left(x,K,\boldsymbol{\xi}\right) & \equiv\sum_{\left\{ \mathbf{w}^{a}\right\} _{a=1}^{y}}\prod_{a=1}^{y}\mathbb{X}_{\boldsymbol{\xi},K}\left(\mathbf{w}^{a}\right)\ \prod_{a<b}^{y}\ind\left(d_{H}\left(\mathbf{w}^{a},\mathbf{w}^{b}\right)=\left\lfloor Nx\right\rfloor \right).\label{eq:Zy}
\end{align}

The summation here is over the $2^{yN}$ spin configurations. We denote
with $\alpha_{c}^{y}(x,K)$ the SAT/UNSAT threshold (if it exists)
in the $N\uparrow\infty$ limit and under the probability distribution
for $\boldsymbol{\xi}$ described in the previous Section. The asymptotic
behavior is captured by the (normalized) local entropy $\phi_{y}$
defined by\footnote{We use a simpler definition compared to ref.~\citep{baldassi2015subdominant}
here, avoiding the explicit use of a reference configuration. The
technical justification for this can be found in ref.~\citep{baldassi2020shaping};
intuitively, the reference is defined implicitly as the barycenter,
and the results are basically equivalent for large $y$.}

\begin{equation}
\phi_{y}\left(x,K,\alpha\right)=\lim_{N\to\infty}\frac{1}{yN}\mathbb{E}_{\boldsymbol{\xi}}\ln\mathcal{Z}_{y}\left(x,K,\boldsymbol{\xi}\right).
\end{equation}

The interpretation of this quantity is as follows. If $\phi_{y}$
is positive, the number of groups of $y$ solutions is exponential.
For any group of $y$ solutions that contributes to the sum in $\mathcal{Z}_{y}$
we can use their barycenter (which will be at distance $r=\frac{1-\sqrt{1-2x}}{2}$
from each of them) as a reference configuration, and in the limit
of large $y$ the sum is dominated by the regions with the highest
density of solutions at distance $r$ from their center, provided
they are evenly distributed. Also in this limit the logarithm of the
density of solutions is computed as $\phi_{y}\left(x,K,\alpha\right)-\phi_{y}\left(x,K,0\right)=\phi_{y}\left(x,K,\alpha\right)-H_{2}\left(\frac{1-\sqrt{1-2x}}{2}\right)$
where $H_{2}\left(r\right)=-r\ln r-\left(1-r\right)\ln\left(1-r\right)$
is the two-state entropy function. If a dense region exists around
a configuration, we should observe a positive $\phi_{y}$ for all
$y$ and for all $x$ in some neighborhood of $0$, and for ultra-dense
regions we should have $\lim_{y\to\infty}\phi_{y}\left(x,K,\alpha\right)=H_{2}\left(\frac{1-\sqrt{1-2x}}{2}\right)-O\left(e^{-\frac{1}{x}}\right)$
for sufficiently small $x$.\footnote{Although these are in principle necessary conditions, and not sufficient,
the latter scenario of a log-density going to $0$ in particular seems
very unlikely in the absence of ultra-dense regions, and indeed when
the matter was investigated numerically for the standard perceptron
model these rare regions were found and their properties were in good
agreement with the theory in a wide range of parameters \citep{baldassi2015subdominant,baldassi2016local}.}

The computation of $\phi_{y}$ can be approached by rigorous techniques
only for small $y$, as discussed in the next sections. In the general
case, for any finite $y$ and in the $y\to\infty$ limit, it can be
carried out at present only using the non-rigorous replica method
of statistical physics of disordered systems. The computations for
this model follow entirely those of ref.~\citep{baldassi2015subdominant}
and are reported in Appendix~\ref{subsec:Limit}.

\begin{figure}
\begin{centering}
\includegraphics[width=0.9\columnwidth]{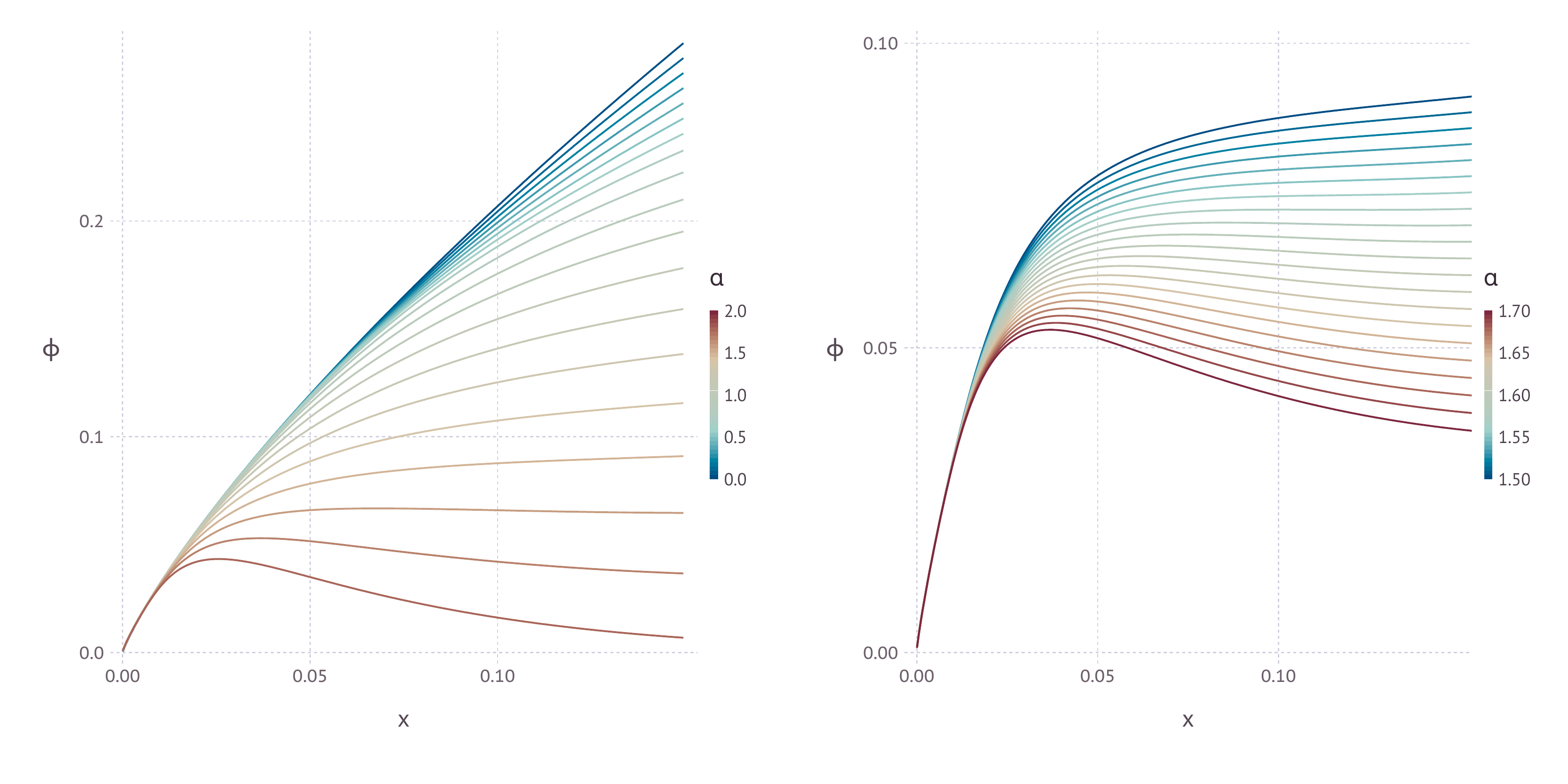}
\par\end{centering}
\caption{\label{fig:clusterone}Plot of free local entropy of eq.~(\ref{eq:Zy})
as a function of the normalized Hamming distance between solutions
$x$, obtained with the replica method using the replica-symmetric
ansatz (see Appendix~\ref{subsec:Limit} for the details). In both
figures the value of the half-width of the channel is $K=1$. (Left)
Curves for $\alpha=0$ up to $\alpha=1.8$ in steps of $0.1$. When
the distance $x$ approaches zero we see that all curves tend to coincide
with the curve for $\alpha=0$, meaning that there exist regions of
solutions that are maximally dense (nearly all configurations are
solutions) in their immediate surroundings. (Right) Zoom on the interval
of values of $\alpha$ where there is a change in monotonicity, which
we interpret as signaling a fragmentation of the dense clusters into
separate pieces. We refine the step of $\alpha$ to $0.01$, and we
find that the change happens for $\alpha_{U}\simeq1.58$.}
\end{figure}

The replica analysis in the $y\to\infty$ limit strongly suggests
the existence of ultra-dense regions of solutions: as shown in fig.~\ref{fig:clusterone},
for $K=1$ and for sufficiently small $x$ the curves for $\alpha$
below the SAT-UNSAT transition, i.e. $\alpha<\alpha_{c}\simeq1.815\dotsc$,
tend to collapse onto the curve for $\alpha=0$, implying that these
regions are maximally dense in their immediate surroundings (nearly
all configurations are solutions in an extensive region centered around
their barycenter). Furthermore, there is a transition at around $\alpha_{U}\simeq1.58$
after which the curves are no longer monotonic. Overall, this is the
same phenomenology that was observed (and confirmed by numerical simulations)
for the standard binary perceptron model in ref.~\citep{baldassi2015subdominant},
and we interpret it in the same way, i.e.~we speculate that ultra-dense
sub-dominant regions of solutions exist, and that the break of monotonicity
at $\alpha_{U}\simeq1.58$ signals a transition\footnote{In ref.~\citep{baldassi2015subdominant} it was shown that some geometric
constraints are violated in a region of $x$ for $\alpha\ge\alpha_{U}$
implying the onset of strong symmetry-breaking effects, with numerical
evidence supporting the switch to a different regime.} between two regimes: one for low $\alpha$ in which the ultra-dense
regions are immersed in a vast connected structure, and one at high
$\alpha$ in which the structure of the dense solutions fragments
into separate regions that are no longer easily accessible.\footnote{It should be noted that in the standard binary perceptron case (i.e.
with sign activation) there is empirical evidence only for the first
scenario of a vast connected structure with ultra-dense regions in
it, while the second scenario of fragmented regions has never been
directly observed at large $N$, arguably due to the intrinsic algorithmic
hardness of finding such regions.}

These results were obtained with the so-called replica-symmetric ansatz,
and they are certainly not exact. However, as in previous studies
\citep{baldassi2015subdominant}, the corrections (which would require
the use of a replica-symmetry-broken ansatz) only become numerically
relevant at relatively large $\alpha$ (e.g. we may expect small corrections
to the value of $\alpha_{U}$, and larger effects close to $\alpha_{c}$),
and they don't affect the qualitative picture, the emerging phenomenology
and its physical interpretation.

\section{Pairs of solutions ($y=2$): rigorous bounds\label{sec:y=00003D2}}

We are able to derive rigorous lower and upper bounds for the existence
of pairs of solutions, i.e. for the $y=2$ case, without resorting
to the replica method. 

The idea of the derivation follows very closely the strategy used
in refs.~\citep{mezard2005clustering,daude2008pairs} for the random
K-SAT problem.

We define a SAT-$x$-pair as a pair of binary weights $\mathbf{w}^{1},\mathbf{w}^{2}\in\left\{ -1,1\right\} ^{N}$,
which are both solutions of the CSP, and whose Hamming distance is
$d_{H}\left(\mathbf{w}^{1},\mathbf{w}^{2}\right)=\left\lfloor Nx\right\rfloor $.
The number of such pairs is $\mathcal{Z}_{y=2}\left(x,K,\boldsymbol{\xi}\right)$,
see eq.~(\ref{eq:Zy}).

\subsection{Upper bound: the first moment method\label{subsec:Upper-bound:-the}}

In this section we are interested in finding an upper-bound (which
depends on $x$) to the critical capacity of pairs of solutions. To
do that we use the upper bound $\mathbb{P}\left[X>0\right]\leq\mathbb{E}\left[X\right]$
that holds when the random variable $X$ is non-negative and integer-valued.
When we apply it to the random variable $\mathcal{Z}_{y=2}$ we get:

\begin{align}
\mathbb{P}\left[\mathcal{Z}_{y=2}\left(x,K,\boldsymbol{\xi}\right)>0\right] & \leq\mathbb{E}\left[\mathcal{Z}_{y=2}\left(x,K,\boldsymbol{\xi}\right)\right]=2^{N}\tbinom{N}{\left\lfloor Nx\right\rfloor }\mathbb{P}\left[v_{1}\in I_{K},v_{2}\in I_{K}\right]^{M}\label{eq:firm}
\end{align}
where we have introduced the two Gaussian random variables $v_{1}$
and $v_{2}$, with $\mathbb{E}\left[v_{1}\right]=\mathbb{E}\left[v_{2}\right]=0,$
$\mathbb{E}\left[v_{1}^{2}\right]=\mathbb{E}\left[v_{2}^{2}\right]=1,$
and covariance
\begin{equation}
\mathbb{E}\left[v_{1}v_{2}\right]=\frac{N-2\left\lfloor Nx\right\rfloor }{N}\underset{N\to+\infty}{\longrightarrow}1-2x.
\end{equation}

Let us consider the normalized logarithm of the first moment,

\begin{equation}
F\left(x,K,\alpha\right)=\lim_{N\to\infty}\frac{1}{N}\ln\mathbb{E}\left[\mathcal{Z}_{y=2}\left(x,K,\boldsymbol{\xi}\right)\right]=\ln2+H_{2}\left(x\right)+\alpha\ln f_{1}\left(x,K\right),
\end{equation}
where as before $H_{2}\left(x\right)=-x\ln x-\left(1-x\right)\ln\left(1-x\right)$
is the two-state entropy function while $f_{1}\left(x,K\right)$ is
defined as follows. Denote with $\Sigma_{2}$ the covariance matrix
of the Gaussian random vector $\vec{v}=\left(v_{1},v_{2}\right)$
whose components have covariance equal to $1-2x$ and variances equal
to one. We define $f_{1}\left(x,K\right)$ as the probability that
this random vector takes values in the box $\left[-K,K\right]^{2}$:

\begin{align}
f_{1}\left(x,K\right) & =\frac{1}{2\pi\left|\Sigma_{2}\right|^{1/2}}\int_{-K}^{K}\int_{-K}^{K}dv_{1}dv_{2}e^{-\vec{v}^{T}\Sigma_{2}^{-1}\vec{v}}\label{eq:f1}\\
 & =\int_{-K}^{K}du_{1}\frac{e^{-u_{1}^{2}/2}}{\sqrt{2\pi}}\int_{\frac{-K-\left(1-2x\right)u_{1}}{2\sqrt{x\left(1-x\right)}}}^{\frac{K-\left(1-2x\right)u_{1}}{2\sqrt{x\left(1-x\right)}}}du_{2}\frac{e^{-u_{2}^{2}/2}}{\sqrt{2\pi}}.\nonumber 
\end{align}
From the inequality (\ref{eq:firm}), $F\left(x,K,\alpha\right)<0$
implies that $\lim_{N\rightarrow\infty}\mathbb{P}\left[\mathcal{Z}_{y=2}\left(x,K,\xi\right)>0\right]=0$.
In turn this provides the upper bound we are seeking:

\begin{ub} For each $K$ and $0<x<1$, and for all $\alpha$ such
that
\begin{equation}
\alpha>\alpha_{UB}\left(x,K\right)\equiv-\frac{\ln2+H_{2}\left(x\right)}{\ln f_{1}\left(x,K\right)}\label{eq:aUB}
\end{equation}
there are no SAT-$x$-pairs w.h.p.\end{ub}. 

Notice that the first moment computation for $Z_{y=2}(x)$ is similar
to the second moment computation for $\mathcal{Z}_{y=1}$ in ref.
\citep{aubin2019storage}: in the former $x$ enters as an external
constraint, in the latter as an order parameter to be optimized.

The upper bound that we obtained for $K=1$ and as a function of $x$
is shown in fig.~\ref{fig:fig2UBLB}. For $x=0$ the upper bound
trivially reduces to the one for a single replica as found in ref.~\citep{aubin2019storage}.
The same happens also for $x=1/2$, as the two constrained replicas
behave as independent systems in the large $N$ limit. 

\subsection{Lower bound: the second moment method \label{subsec:Lower-bound:-the}}

\global\long\def\mb{\mathbf{w}}%
\global\long\def\mbt{\tilde{\mathbf{w}}}%

We compute the lower bound to the critical capacity using the second
moment method, which is a direct consequence of the Cauchy-Schwarz
inequality:
\begin{lem}
[Second moment method]If $X$ is a non-negative random variable,
then

\begin{equation}
\mathbb{P}\left[X>0\right]\geq\frac{\mathbb{E}\left[X\right]^{2}}{\mathbb{E}\left[X^{2}\right]}.\label{eq:secmom}
\end{equation}
\end{lem}
From the results of section \ref{subsec:Upper-bound:-the} we have 

\begin{equation}
\mathbb{E}\left[\mathcal{Z}_{y=2}\left(x,K,\boldsymbol{\xi}\right)\right]=2^{N}\binom{N}{\left\lfloor Nx\right\rfloor }f_{1}\left(\frac{\left\lfloor Nx\right\rfloor }{N},K\right)^{M},
\end{equation}
where $f_{1}\left(x,K\right)$ is defined like in eq.~(\ref{eq:f1}).
The second moment of the random variable $\mathcal{Z}_{y=2}$ follows
from simple combinatorics and reads

\begin{align}
\mathbb{E}\left[\mathcal{Z}_{y=2}^{2}\left(x,K,\boldsymbol{\xi}\right)\right] & =\sum_{\left\{ \mathbf{w}^{1}\right\} }\sum_{\left\{ \mathbf{w}^{2}\right\} }\sum_{\left\{ \mathbf{\tilde{w}}^{1}\right\} }\sum_{\left\{ \mathbf{\tilde{w}}^{2}\right\} }\ind\left(d_{H}\left(\mathbf{w}^{1},\mathbf{w}^{2}\right)=\left\lfloor Nx\right\rfloor \right)\ \ind\left(d_{H}\left(\mathbf{\tilde{w}}^{1},\mathbf{\tilde{w}}^{2}\right)=\left\lfloor Nx\right\rfloor \right)\times\nonumber \\
 & \quad\times\prod_{\mu=1}^{M}\mathbb{E}\left[\ind\left(w^{1}\cdot\xi^{\mu}\in I_{K}\right)\ \ind\left(w^{2}\cdot\xi^{\mu}\in I_{K}\right)\ \ind\left(\tilde{w}^{1}\cdot\xi^{\mu}\in I_{K}\right)\ \ind\left(\tilde{w}^{2}\cdot\xi^{\mu}\in I_{K}\right)\right]\label{eq:ExpZ2}\\
 & =2^{N}\sum_{\mathbf{a}\in V_{N,x}\cap\left\{ 0,1/N,2/N,\ldots,1\right\} ^{8}}\frac{N!}{\prod_{i=0}^{7}\left(Na_{i}\right)!}f_{2}\left(\mathbf{a},x,K\right)^{M},\nonumber 
\end{align}
where we have adopted the following conventions.
\begin{itemize}
\item $\mathbf{\mathbf{a}}$ is an $8$-component vector giving the proportion
of each type of quadruplets $\left(w_{i}^{1},w_{i}^{2},\tilde{w}_{i}^{1},\tilde{w}_{i}^{2}\right)$
as described in the table below, where we have arbitrarily (but without
loss of generality) fixed $\mathbf{w}^{1}$ to $\left(1,\ldots,1\right)$.
Fixing the vector $\mathbf{\mathbf{a}}$ entails fixing all the possible
overlaps between the vectors $w^{1},w^{2},\tilde{w}^{1}$ and $\tilde{w}^{2}$
and consequently the covariances of the random variables $z_{1}:=w^{1}\cdot\xi$,
$z_{2}:=w^{2}\cdot\xi$, $\tilde{z}_{1}:=\tilde{w}^{1}\cdot\xi$ and
$\tilde{z}_{2}:=\tilde{w}^{2}\cdot\xi$ with $\xi_{i}\sim\mathcal{N}\left(0,1/N\right)$
i.i.d. These covariances as functions of $\mathbf{a}$ are made explicit
in eq.~(\ref{eq:alloverlaps}).
\end{itemize}
\begin{center}
\begin{tabular}{ccccccccc}
 & $a_{0}$ & $a_{1}$ & $a_{2}$ & $a_{3}$ & $a_{4}$ & $a_{5}$ & $a_{6}$ & $a_{7}$\tabularnewline
\hline 
$w_{i}^{1}$  & + & + & + & + & + & + & + & +\tabularnewline
$w_{i}^{2}$ & + & + & + & + & $-$ & $-$ & $-$ & $-$\tabularnewline
$\tilde{w}_{i}^{1}$ & + & + & $-$ & $-$ & + & + & $-$ & $-$\tabularnewline
$\tilde{w}_{i}^{2}$ & + & $-$ & + & $-$ & + & $-$ & + & $-$ \tabularnewline
\end{tabular}
\par\end{center}
\begin{itemize}
\item $f_{2}\left(\mathbf{a},x,K\right)$ has the expression
\begin{align*}
f_{2}\left(\mathbf{\mathbf{a}},x,K\right) & =\mathbb{P}\left[z_{1}\in I_{K},z_{2}\in I_{K},\tilde{z}_{1}\in I_{K},\tilde{z}_{2}\in I_{K}\right].
\end{align*}
where $\mathbf{z}^{T}\coloneqq\left(z_{1},z_{2},\tilde{z}_{1},\tilde{z}_{2}\right)$
is a $4$-dimensional Gaussian vector, with the following set of covariances:
\\
\begin{equation}
\Sigma=\left(\begin{array}{cccc}
1 & q_{1} & q_{01} & q_{02}\\
q_{1} & 1 & q_{03} & q_{04}\\
q_{01} & q_{03} & 1 & q_{1}\\
q_{02} & q_{04} & q_{1} & 1
\end{array}\right)\quad\textrm{where}\quad\left\{ \begin{array}{l}
q_{1}=1-2\frac{\left\lfloor Nx\right\rfloor }{N}\\
q_{01}=1-2\left(a_{2}+a_{3}+a_{6}+a_{7}\right)\\
q_{02}=1-2\left(a_{1}+a_{3}+a_{5}+a_{7}\right)\\
q_{03}=1-2\left(a_{2}+a_{3}+a_{4}+a_{5}\right)\\
q_{04}=1-2\left(a_{1}+a_{3}+a_{4}+a_{6}\right)
\end{array}\right..\label{eq:alloverlaps}
\end{equation}
Therefore $f_{2}\left(\mathbf{\mathbf{a}},x,K\right)$ can be simply
written as the following Gaussian integral
\begin{equation}
f_{2}\left(\mathbf{a},x,K\right)=\int_{I_{K}^{4}}dz_{1}dz_{2}d\tilde{z}_{1}d\tilde{z}_{2}\frac{1}{\left(2\pi\right)^{2}\left|\Sigma\right|^{1/2}}e^{-\frac{1}{2}\mathbf{z}^{T}\Sigma^{-1}\text{\ensuremath{\mathbf{z}}}}.
\end{equation}
 
\item The set $V_{N,x}\subset\left[0,1\right]^{8}$ is a simplex specified
by: 
\begin{equation}
\left\{ \begin{array}{l}
\lfloor N\left(a_{4}+a_{5}+a_{6}+a_{7}\right)\rfloor=\lfloor Nx\rfloor\\
\lfloor N\left(a_{1}+a_{2}+a_{5}+a_{6}\right)\rfloor=\lfloor Nx\rfloor\\
\sum_{i=0}^{7}a_{i}=1
\end{array}\right..\label{eq:eqsympN}
\end{equation}
These three conditions correspond to the normalization of the proportions
and to the enforcement of the conditions $d_{\mb^{1}\mb^{2}}=\lfloor Nx\rfloor$,
$d_{\mbt^{1}\mbt^{2}}=\lfloor Nx\rfloor$. When $N\to\infty$, $V_{x}=\bigcap_{N\in\mathbb{N}}V_{N,x}$
defines a five-dimensional simplex described by the three hyperplanes:
\begin{equation}
\left\{ \begin{array}{l}
a_{4}+a_{5}+a_{6}+a_{7}=x\\
a_{1}+a_{2}+a_{5}+a_{6}=x\\
\sum_{i=0}^{7}a_{i}=1
\end{array}\right..\label{eq:const}
\end{equation}
\end{itemize}
In order to yield an asymptotic estimate of $\mathbb{E}\left[\mathcal{Z}_{y=2}^{2}\right]$
we first use the following known result, which comes from the approximation
of integrals by sums (proof in Appendix~\ref{subsec:A-Proof-of-Lemma-2}):
\begin{lem}
\label{Lemma02}Let $\psi\left(\mathbf{a}\right)$ be a real, positive,
continuous function of $\mathbf{a}$, and let $V_{N,x}$, $V_{x}$
be as defined above. Then for any given $x$ there exists a constant
$C_{0}$ such that for sufficiently large $N$:\footnote{Here and below this $8$-dimensional integration is to be understood
as being performed with a uniform measure in the $5$-dimensional
subspace $V_{x}$, i.e.~$\int_{V_{x}}d\mathbf{a}\equiv\int_{[0,1]^{8}}d\mathbf{a}\,\delta\left(a_{4}+a_{5}+a_{6}+a_{7}-x\right)\delta\left(a_{1}+a_{2}+a_{5}+a_{6}-x\right)\delta\left(\sum_{i=0}^{7}a_{i}-1\right)$,
where $\delta$ is a Dirac delta, cf.~eq.~(\ref{eq:const}).}
\begin{equation}
\sum_{\mathbf{a}\in V_{N,x}\cap\left\{ 0,1/N,2/N,\ldots,1\right\} ^{8}}\frac{N!}{\prod_{i=0}^{7}\left(Na_{i}\right)!}\psi\left(\mathbf{a}\right)^{N}\leq C_{0}N^{3/2}\int_{V_{x}}d\mathbf{a}\ e^{N\left[H_{8}\left(\mathbf{a}\right)+\ln\psi\left(\mathbf{a}\right)\right]},\label{eq:discont}
\end{equation}
where $H_{8}\left(\mathbf{a}\right)=-\sum_{i=0}^{7}a_{i}\ln a_{i}$.
\end{lem}
The bound for the second moment then reads:

\begin{equation}
\mathbb{E}\left[\mathcal{Z}_{y=2}^{2}\left(x,K,\boldsymbol{\xi}\right)\right]\leq C_{0}N^{3/2}\int_{V_{x}}d\mathbf{a}\ e^{N\left[\ln2+H_{8}\left(\mathbf{a}\right)+\alpha\ln f_{2}\left(\mathbf{a},x,K\right)\right]},\label{eq:boundsm}
\end{equation}
which is obtained from substitution of eq.~(\ref{eq:ExpZ2}) into
Lemma \ref{Lemma02}. The number of components of the vector $\mathbf{a}$
is eight, but we can reduce their number to five with a change of
variables and rewrite the integral in a particularly simple form where
$f_{2}$ just depends on four of them. This is done in Appendix~\ref{subsec:Change-of-integration}.
Here we give just the final expression where the new integration variables
are $\eta$ (a scalar) and $\vec{q}_{0}=\left(q_{01},q_{02},q_{03},q_{04}\right)$.
The bound becomes

\begin{equation}
\mathbb{E}\left[\mathcal{Z}_{y=2}^{2}\left(x,K,\boldsymbol{\xi}\right)\right]\leq C_{0}N^{3/2}\int_{\tilde{V}_{x}}d\vec{q}_{0}\ d\eta\ e^{N\left[\ln2+H_{8}\left(\vec{q}_{0},\eta,x\right)+\alpha\ln f_{2}\left(\vec{q}_{0},x,K\right)\right]},\label{eq:boundsm3}
\end{equation}
where:
\begin{itemize}
\item $f_{2}\left(\vec{q}_{0},x,K\right)$ has the expression
\[
f_{2}\left(\vec{q}_{0},x,K\right)=\int_{I_{K}^{4}}dz_{1}dz_{2}d\tilde{z}_{1}d\tilde{z}_{2}\frac{1}{\left(2\pi\right)^{2}\left|\Sigma\right|^{1/2}}e^{-\frac{1}{2}\mathbf{z}^{T}\Sigma^{-1}\text{\ensuremath{\mathbf{z}}}},
\]
where $\Sigma$ is the covariance matrix of eq.~(\ref{eq:alloverlaps})
with $q_{1}=1-2x$ and where the components of $\vec{q}_{0}$ are
considered as independent variables.
\item $H_{8}\left(\vec{q}_{0},\eta,x\right)$ is defined as the Shannon
entropy of a probability mass function with masses corresponding to
the components of the following vector:
\begin{equation}
\left(\begin{array}{c}
\frac{1}{4}\left(q_{02}+q_{03}+2-4x\right)+\eta\\
\frac{1}{4}\left(q_{01}-q_{02}+2x\right)-\eta\\
\frac{1}{4}\left(-q_{03}+q_{04}+2x\right)-\eta\\
\frac{1}{4}\left(2-q_{01}-q_{04}-4x\right)+\eta\\
\frac{1}{4}\left(q_{01}-q_{03}+2x\right)-\eta\\
\eta\\
\frac{1}{4}\left(-q_{01}+q_{02}+q_{03}-q_{04}\right)+\eta\\
\frac{1}{4}\left(-q_{02}+q_{04}+2x\right)-\eta
\end{array}\right);\label{eq:vectorH8}
\end{equation}
\item $\tilde{V}_{x}$ is the new domain of integration specified by the
inequalities
\end{itemize}
\begin{equation}
\left\{ \begin{array}{l}
\frac{1}{4}\left(q_{01}-q_{02}+2x-4\right)\leq\eta\leq\frac{1}{4}\left(q_{01}-q_{02}+2x\right)\\
\frac{1}{4}\left(-q_{03}+q_{04}+2x-4\right)\leq\eta\leq\frac{1}{4}\left(-q_{03}+q_{04}+2x\right)\\
\frac{1}{4}\left(q_{01}+q_{04}+4x-2\right)\leq\eta\leq\frac{1}{4}\left(q_{01}+q_{04}+4x+2\right)\\
\frac{1}{4}\left(q_{01}-q_{03}+2x-4\right)\leq\eta\leq\frac{1}{4}\left(q_{01}-q_{03}+2x\right)\\
0\leq\eta\leq1\\
\frac{1}{4}\left(q_{01}-q_{02}-q_{03}+q_{04}\right)\leq\eta\\
\frac{1}{4}\left(-q_{02}+q_{04}+2x-4\right)\leq\eta\leq\frac{1}{4}\left(-q_{02}+q_{04}+2x\right)\\
\frac{1}{4}\left(-q_{02}-q_{03}+4x-2\right)\leq\eta
\end{array}\right.,\label{eq:syseqris2}
\end{equation}

some of which are already contained in eq.~(\ref{eq:vectorH8}). 
\begin{prop}
\label{phi-secmombound}For each $K,x$, define:
\begin{equation}
\Phi_{x,K,\alpha}\left(\vec{q}_{0},\eta\right)=H_{8}\left(\vec{q}_{0},\eta,x\right)-\ln2-2H_{2}\left(x\right)+\alpha\ln f_{2}\left(\vec{q}_{0},x,K\right)-2\alpha\ln f_{1}\left(x,K\right).\label{defphi}
\end{equation}
and let $\left(\vec{q}_{0}^{M},\eta^{M}\right)\in\tilde{V}_{x}$ be
the global maximum of $\Phi_{x,K,\alpha}$ restricted to $\tilde{V}_{x}$.
Then there exists a $x,K$-dependent constant $C>0$ such that, for
$N$ sufficiently large, 
\begin{equation}
\frac{\mathbb{E}\left[\mathcal{Z}_{y=2}\left(x,K,\boldsymbol{\xi}\right)\right]^{2}}{\mathbb{E}\left[\mathcal{Z}_{y=2}^{2}\left(x,K,\boldsymbol{\xi}\right)\right]}\geq C\exp\left(-N\Phi_{x,K,\alpha}\left(\vec{q}_{0}^{M},\eta^{M}\right)\right).\label{cond}
\end{equation}
\end{prop}
\begin{proof}
Applying Laplace method to the integral in eq.~(\ref{eq:boundsm3}),
for some constant $C_{1}$ and for $N$ large enough we obtain 
\begin{equation}
\mathbb{E}\left[\mathcal{Z}_{y=2}^{2}\left(x,K,\boldsymbol{\xi}\right)\right]\leq C_{1}N^{-1}e^{N\left[\ln2+H_{8}\left(\vec{q}_{0}^{M},\eta,x\right)+\alpha\ln f_{2}\left(\vec{q}_{0}^{M},x,K\right)\right]},
\end{equation}

where the factor $N^{-1}=N^{\frac{3}{2}-\frac{5}{2}}$ stems from
the Gaussian fluctuations around the $5$-dimensional saddle point.
For the first moment instead, a simple application of Stirling formula
to eq.~(\ref{eq:firm}) leads, for some constant $c_{1}$ and $N$
large enough, to 

\begin{equation}
\mathbb{E}\left[\mathcal{Z}_{y=2}\left(x,K,\boldsymbol{\xi}\right)\right]^{2}\geq c_{1}N^{-1}e^{2N\left[\ln2+H_{2}\left(x\right)+\alpha\ln f_{1}\left(x,K\right)\right]}.
\end{equation}

Combining the two expressions, the proposition follows.
\end{proof}
Given that $\Phi_{x,K,\alpha}\left(\vec{q}_{0}^{M},\eta^{M}\right)\ge0$,
the second moment method gives a useful bound just when $\Phi_{x,K,\alpha}\left(\vec{q}_{0}^{M},\eta^{M}\right)=0$.
If instead $\Phi_{x,K,\alpha}\left(\vec{q}_{0}^{M},\eta^{M}\right)>0$,
the probability is bounded above zero (included) and the bound is
non-informative. 

For a particular point $\left(\vec{q}_{0}^{\star},\eta^{\star}\right)\in\tilde{V}_{x}$,
which can be interpreted intuitively as capturing the situation where
the two pairs of solutions are uncorrelated, we have that $\Phi_{x,K,\alpha}\left(\vec{q}_{0}^{\star},\eta^{\star}\right)=0$
for all values of $\alpha$. This point $\left(\vec{q}_{0}^{\star},\eta^{\star}\right)$
is specified by the following equations,
\begin{equation}
q_{01}^{\star}=0,\ q_{02}^{\star}=0,\ q_{03}^{\star}=0,\ q_{04}^{\star}=0,\ \eta^{\star}=\frac{x^{2}}{2}.
\end{equation}
In that case, we have the following properties:
\begin{itemize}
\item $H_{8}\left(\vec{q}_{0}^{\star},\eta^{\star},x\right)=\ln2+2H_{2}\left(x\right)$, 
\item $f_{2}\left(\vec{q}_{0}^{\star},x,K\right)=f_{1}\left(x,K\right)^{2}$.
\end{itemize}
Therefore, $\alpha_{LB}$ is the largest value of $\alpha$ such that
$\left(\vec{q}_{0}^{\star},\eta^{\star}\right)$ is a global maximum,
i.e.~such that there exists no $\left(\vec{q}_{0},\eta\right)\in\tilde{V}_{x}$
with $\Phi_{x,K,\alpha}\left(\vec{q}_{0},\eta\right)>0$. In particular,
for $\alpha=0$ the second moment bound holds (proof in Appendix~\ref{subsec:Proof-of-eq}):
\begin{equation}
\Phi_{x,K,\alpha=0}\left(\vec{q}_{0},\eta\right)=H_{8}\left(\vec{q}_{0},\eta,x\right)-\ln2-2H_{2}\left(x\right)\leq0\quad\forall\left(\vec{q}_{0},\eta\right)\in\tilde{V}_{x}.\label{eq:ineqAlpha0}
\end{equation}

Now, let us split $\tilde{V}_{x}$ in the following way: 
\[
\tilde{V}_{x}^{+}:=\left\{ \left(\vec{q}_{0},\eta\right)\in\tilde{V}_{x}\mid f_{2}\left(\vec{q}_{0},x,K\right)>f_{1}^{2}\left(x,K\right)\right\} \quad\text{and\ensuremath{\quad}\ensuremath{\tilde{V}_{x}^{-}}:=\ensuremath{\left\{  \left(\vec{q}_{0},\eta\right)\in\tilde{V}_{x}\mid f_{2}\left(\vec{q}_{0},x,K\right)\leq f_{1}^{2}\left(x,K\right)\right\} } }.
\]
 It follows that for all $\left(\vec{q}_{0},\eta\right)\in\tilde{V}_{x}^{-}$
and $\alpha>0$ we have 
\begin{align*}
\Phi_{x,K,\alpha}\left(\vec{q}_{0},\eta\right) & \leq\Phi_{x,K,\alpha=0}\left(\vec{q}_{0},\eta\right)\leq0.
\end{align*}
As already discussed, $\alpha_{LB}$ is the largest value of $\alpha$
such that 
\[
\max_{\left(\vec{q}_{0},\eta\right)\in\tilde{V}_{x}}\Phi_{x,K,\alpha}\left(\vec{q}_{0},\eta\right)=0.
\]
From the previous observation 
\[
\max_{\left(\vec{q}_{0},\eta\right)\in\tilde{V}_{x}}\Phi_{x,K,\alpha}\left(\vec{q}_{0},\eta\right)=\sup_{\left(\vec{q}_{0},\eta\right)\in\tilde{V}_{x}^{+}}\Phi_{x,K,\alpha}\left(\vec{q}_{0},\eta\right),
\]
and therefore $\alpha_{LB}$ is the largest value of $\alpha$ such
that 
\[
\sup_{\left(\vec{q}_{0},\eta\right)\in\tilde{V}_{x}^{+}}\Phi_{x,K,\alpha}\left(\vec{q}_{0},\eta\right)=0.
\]
Then, $\alpha_{LB}$ is the largest value of $\alpha$ such that there
exists no $\left(\vec{q}_{0},\eta\right)\in\tilde{V}_{x}^{+}$ with
$\Phi_{x,K,\alpha}\left(\vec{q}_{0},\eta\right)>0$, which is true
if and only if
\begin{equation}
H_{8}\left(\vec{q}_{0},\eta,x\right)-\ln2-2H_{2}\left(x\right)+\alpha\ln f_{2}\left(\vec{q}_{0},x,K\right)-2\alpha\ln f_{1}\left(x,K\right)\leq0\quad\forall\ \left(\vec{q}_{0},\eta\right)\in\tilde{V}_{x}^{+},\quad\forall\ \alpha\leq\alpha_{LB}.\label{eq:condLB}
\end{equation}
Therefore, eq.~(\ref{eq:condLB}) implies that for $\alpha\leq\alpha_{LB}$
the following condition must hold as well:

\begin{equation}
\alpha\leq\frac{\ln2+2H_{2}\left(x\right)-H_{8}\left(\vec{q}_{0},\eta,x\right)}{\ln f_{2}\left(\vec{q}_{0},x,K\right)-2\ln f_{1}\left(x,K\right)}\quad\forall\ \left(\vec{q}_{0},\eta\right)\in\tilde{V}_{x}^{+}.
\end{equation}
We obtain the following result:

\begin{lb}For each $K$ and $0<x<1$, and for all $\alpha$ such
that 

\begin{equation}
\alpha<\alpha_{LB}\left(x,K\right)\equiv\inf_{\left(\vec{q}_{0},\eta\right)\in\tilde{V}_{x}^{+}}\frac{\ln2+2H_{2}\left(x\right)-H_{8}\left(\vec{q}_{0},\eta,x\right)}{\ln f_{2}\left(\vec{q}_{0},x,K\right)-2\ln f_{1}\left(x,K\right)}\label{eq:aLB}
\end{equation}
we have that there is a positive probability of finding SAT-$x$-pairs
of solutions, namely

\begin{equation}
\liminf\limits _{N\rightarrow\infty}\mathbb{P}\left[\mathcal{Z}_{y=2}\left(x,K,\boldsymbol{\xi}\right)>0\right]>0.
\end{equation}
\end{lb}

The optimization can be simplified further by slicing the set $\tilde{V}_{x}^{+}$
in the two ``directions'' $\vec{q}_{0}$ and $\eta$. We define
a $\vec{q}_{0}$-slice as $\left(\tilde{V}_{x}^{+}\right)_{\vec{q}_{0}}:=\left\{ \eta\mid\left(\vec{q}_{0},\eta\right)\in\tilde{V}_{x}^{+}\right\} $
and the natural projection of the set $\tilde{V}_{x}^{+}$ on the
$\vec{q}_{0}$-subspace as $\pi_{\vec{q}_{0}}\left(\tilde{V}_{x}^{+}\right)=\left\{ \vec{q}_{0}\mid\exists\,\eta\ \textrm{s.t. }\left(\vec{q}_{0},\eta\right)\in\tilde{V}_{x}^{+}\right\} $.
With this notation, eq.~(\ref{eq:aLB}) becomes:

\begin{equation}
\alpha_{LB}\left(x,K\right)=\inf_{\vec{q}_{0}\in\pi_{\vec{q}_{0}}\left(\tilde{V}_{x}^{+}\right)}\frac{\ln2+2H_{2}\left(x\right)-\sup_{\eta\in\left(\tilde{V}_{x}^{+}\right)_{\vec{q}_{0}}}H_{8}\left(\vec{q}_{0},\eta,x\right)}{\ln f_{2}\left(\vec{q}_{0},x,K\right)-2\ln f_{1}\left(x,K\right)}.\label{eq:aLB2}
\end{equation}

The optimization in $\eta$ is easy because the function $H_{8}\left(\vec{q}_{0},\eta\right)$
is concave in $\eta$ for each $\vec{q}_{0}$. This is not necessarily
true for the optimization in $\vec{q}_{0}$. In fact, it is crucial
that we find the global optimum of the objective function because
this gives the correct value for the lower bound. To this purpose
we have devised two computational strategies. First we evaluated the
objective function on a $4$-dimensional grid with increasing number
of points. Then we have also implemented a simple gradient descent
starting from the points of the grid. The different strategies are
discussed in Appendix~\ref{sec:Numerical-optimization}.

\begin{figure}
\begin{centering}
\includegraphics[width=1\columnwidth]{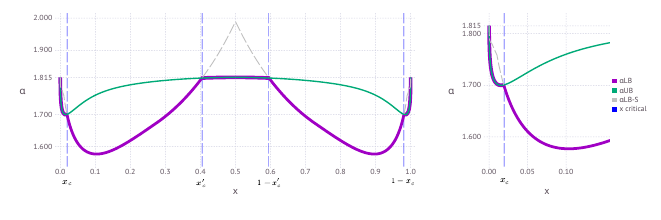}
\par\end{centering}
\caption{\label{fig:fig2UBLB}Lower and upper bounds for the RBP with $K=1$.
(Left) Lower and upper bounds on the whole range $x\in\left[0,1\right]$.
These bounds are symmetric around the vertical axis that passes by
$x=0.5$. In correspondence of the $SB$ solution, the lower bound
prediction from the $S$ point (gray line) is larger than the upper
bound and therefore patently wrong. This is what happens in the regions
$x\protect\leq x_{c}$, $x\protect\geq1-x_{c}$ and $x_{c}^{\prime}\protect\leq x\protect\leq1-x_{c}^{\prime}$,
where the two critical values $x_{c}\simeq0.195\dotsc$ and $x_{c}^{\prime}\simeq0.405\dotsc$
are highlighted by the blue vertical lines on the left of the symmetry
axis. In the regions $x_{c}\protect\leq x\protect\leq x_{c}^{\prime}$
and $1-x_{c}^{\prime}\protect\leq x\protect\leq1-x_{c}$, there is
a gap between the lower (purple line) and the upper bound (green line)
where the $S$ solution is indeed valid. (Right) Zoom of the figure
on the left, in the region around $x_{c}$. Here, for $x\protect\leq x_{c}$
the $S$ solution fails. This is evident from the fact that the symmetric
lower bound becomes bigger than the upper bound (gray line). In this
region instead, the true lower bound perfectly matches the upper bound
since the optimum of eq.~(\ref{eq:y4fm}) is in correspondence of
the $SB$ solution. }
\end{figure}

The bounds that we obtain in fig.~\ref{fig:fig2UBLB} are symmetric
around the value $x=0.5$ and there are two critical values $x_{c},x_{c}^{\prime}\in\left[0,0.5\right]$
(plus the symmetric ones, $1-x_{c}$ and $1-x_{c}^{\prime}$) that
delimit regions characterized by two different phases. For values
of $x$ such that $x_{c}\leq x\leq x_{c}^{\prime}$ or $1-x_{c}^{\prime}\leq x\leq1-x_{c}$
all four entries of $\vec{q}_{0}$ take the same value. We use the
subscript $S$ to denote this kind of solution. Instead for $x\leq x_{c}$,
$x\geq1-x_{c}$ and $x_{c}^{\prime}\leq x\leq1-x_{c}^{\prime}$, this
symmetry is broken and the optimum is achieved on a different point
that we call Symmetry Broken ($SB$) solution. This point has the
property that the two pairs of binary vectors of solutions $\left(\mathbf{w}_{1},\mathbf{w}_{2}\right)$
and $\left(\mathbf{\tilde{w}}_{1},\mathbf{\tilde{w}}_{2}\right)$
coincide, as can be seen from the structure of the covariance matrix.
We report below the covariance structure of the two solutions, where
we adopted the convention $q_{1}:=1-2x$. The symmetric covariance
matrix is the following: 

\begin{equation}
\Sigma_{S}=\left(\begin{array}{cccc}
1 & q_{1} & q_{0} & q_{0}\\
q_{1} & 1 & q_{0} & q_{0}\\
q_{0} & q_{0} & 1 & q_{1}\\
q_{0} & q_{0} & q_{1} & 1
\end{array}\right),
\end{equation}
while the one corresponding to point $SB$ is the following:

\begin{equation}
\Sigma_{SB}=\left(\begin{array}{cccc}
1 & q_{1} & 1 & q_{1}\\
q_{1} & 1 & q_{1} & 1\\
1 & q_{1} & 1 & q_{1}\\
q_{1} & 1 & q_{1} & 1
\end{array}\right).\label{eq:pointA}
\end{equation}
This is a degenerate covariance matrix, and in correspondence of the
$SB$ solution it follows from the previous equations that the lower
bound and the upper bound coincide.

The physical meaning of what happens is qualitatively different for
these two phases. Let us take $x<x_{c}$, where the bounds are tight,
and let's start with low $\alpha$ and progressively increase it.
In this regime the typical overlap between pairs of solutions is zero,
i.e. the two pairs of solutions are independent and there is a positive
probability of finding SAT-$x$-pairs since we are below $\alpha_{LB}$.
When we reach $\alpha=\alpha_{LB}=\alpha_{UB}$ there is a transition
to a regime where w.h.p. there exists no pair of solutions to the
problem. When this happens the point $\left(\vec{q}_{0}^{M},\eta^{M}\right)\in\tilde{V}_{x}$
that optimizes (\ref{defphi}) is the $SB$ point. For $x_{c}<x<x_{c}^{\prime}$,
the bounds are no longer tight and we can only identify  a region
between the two bounds where a SAT/UNSAT transition occurs. Again,
for $x_{c}^{\prime}<x\leq0.5$ the bounds are tight. For $x>0.5$
the behavior is symmetric to the one that we have just described.

\section{Multiplets of solutions ($y>2$)\label{sec:Caseygen}}

In the previous section we were able to derive rigorous expressions
for the upper bound $\alpha_{UB}\left(x\right)$, in eq.~(\ref{eq:aUB}),
and the lower bound $\alpha_{LB}\left(x\right)$, in eq.~(\ref{eq:aLB2}),
obtained by first and second moment calculations, such that w.h.p.
no pairs of solutions at distance $x$ exist for load $\alpha>\alpha_{UB}\left(x\right)$
and at least one pair exists for $\alpha<\alpha_{LB}\left(x\right)$.
It would be then natural to try to generalize the derivation to sets
of $y$ solutions at pairwise distance $x$ (multiplets) and in particular
asses the existence of a small $\alpha$ regime where such sets can
be found for any value of $y$ and for small enough $x$. This result
would rigorously confirm the existence of a dense region of solutions
as derived in sec.~\ref{sec:large-dev}, which in turn has been non-rigorously
advocated as a necessary condition for the existence of efficient
learning algorithms \citep{baldassi2015subdominant}.

Unfortunately, it is technically unfeasible to carry out the rigorous
derivation for $y>2$ as we have done above for the case $y=2$. Therefore,
in this section, after giving an rigorous expression for the first
moment upper bound limited to the cases $y=3$ and $y=4$, we will
derive compact expressions for the first and second moment bound using
non-rigorous field theoretical calculations and a replica symmetric
ansatz. We find that the non-rigorous results match the rigorous ones
when available, although we expect the prediction to break down at
large values of $y$ due to replica symmetry breaking effects (see
the discussion in the Introduction).

\subsection{Rigorous first moment upper bounds}

In the following we derive the rigorous expressions for the first
moment bound in two additional cases: the existence of triplets and
quadruplets of solutions at fixed pairwise distance $x$.

\subsubsection{Triplets ($y=3$)}

Let us define the symbol $\cong$ as equivalence up to sub-exponential
terms as $N\to\infty$, that is for any two sequences $\left(a_{N}\right)_{N}$
and $\left(b_{N}\right)_{N}$ we write $a_{N}\cong b_{N}$ iff $\lim_{N\to+\infty}\frac{\ln a_{N}}{\ln b_{N}}=1$.
The first moment of the triplets partition function has the following
asymptotic expression:

\begin{align}
\mathbb{E}\left[\mathcal{Z}_{y=3}\left(x,K,\boldsymbol{\xi}\right)\right] & \cong2^{N}\tbinom{N}{\frac{Nx}{2},\frac{Nx}{2},\frac{Nx}{2},N\left(1-\frac{3}{2}x\right)}\mathbb{P}\left[v_{1}\in I_{K},v_{2}\in I_{K},v_{3}\in I_{K}\right]^{M}\\
 & \cong e^{N\left(\ln\left(2\right)+H_{4}\left(x\right)+\alpha\ln f_{1}^{y=3}\left(x,K\right)\right)},\nonumber 
\end{align}
where $H_{4}\left(x\right)=-\frac{3}{2}x\ln\left(\frac{x}{2}\right)-\left(1-\frac{3}{2}x\right)\ln\left(1-\frac{3}{2}x\right)$
and we get the geometric condition $0<x<\frac{2}{3}$, and $f_{1}^{y=3}\left(x,K\right)$
is the probability that a zero mean Gaussian random vector $\vec{v}_{3}=\left(v_{1},v_{2},v_{3}\right)$,
whose covariance matrix $\Sigma_{3}$ has ones on the diagonal and
$1-2x$ off-diagonal, takes values in the box $\left[-K,K\right]^{3}$,
that is

\begin{align}
f_{1}^{y=3}\left(x,K\right) & =\frac{1}{\left(2\pi\right)^{\frac{3}{2}}\left|\Sigma_{3}\right|^{1/2}}\int_{\left[-K,K\right]^{3}}dv_{1}dv_{2}dv_{3}\ e^{-\vec{v}_{3}^{T}\Sigma_{3}^{-1}\vec{v}_{3}}.
\end{align}

An equivalent argument to the case $y=2$ gives the following upper
bound for the existence of clusters of three solutions:

\begin{equation}
\alpha_{UB}^{y=3}\left(x,K\right)=-\frac{\ln2+H_{4}\left(x\right)}{\ln f_{1}^{y=3}\left(x,K\right)}.
\end{equation}

\subsubsection{Quadruplets ($y=4$)}

For quadruplets of solutions, we have

\begin{align}
\mathbb{E}\left[\mathcal{Z}_{y=4}\left(x,K,\boldsymbol{\xi}\right)\right] & \cong2^{N}\sum_{\mathbf{a}\in V_{N,x}^{y=4}\cap\left\{ 0,1/N,2/N,\ldots,1\right\} ^{8}}\frac{N!}{\prod_{i=0}^{7}\left(Na_{i}\right)!}\left[f_{1}^{y=4}\left(x,K\right)\right]^{M},\label{eq:y4fm}
\end{align}
where:
\begin{itemize}
\item In complete analogy with the previous case $f_{1}^{y=4}\left(x,K\right)$
is the probability that a zero mean Gaussian random vector $\vec{v}_{4}=\left(v_{1},v_{2},v_{3},v_{4}\right)$,
whose covariance matrix $\Sigma_{4}$ has ones on the diagonal and
$1-2x$ off-diagonal, takes values in the box $\left[-K,K\right]^{4}$,
that is
\begin{equation}
f_{1}^{y=4}\left(x,K\right)=\frac{1}{\left(2\pi\right)^{2}\left|\Sigma_{4}\right|^{1/2}}\int_{\left[-K,K\right]^{4}}d\vec{v}_{4}\ e^{-\vec{v}_{4}^{T}\Sigma_{4}^{-1}\vec{v}_{4}}.
\end{equation}
\item The summation is restricted to the set $V_{N,x}^{y=4}\subseteq\left[0,1\right]^{8}$,
specified by: 
\begin{equation}
\left\{ \begin{array}{l}
\lfloor N\left(a_{4}+a_{5}+a_{6}+a_{7}\right)\rfloor=\lfloor Nx\rfloor\\
\lfloor N\left(a_{1}+a_{2}+a_{5}+a_{6}\right)\rfloor=\lfloor Nx\rfloor\\
\lfloor N\left(a_{2}+a_{3}+a_{6}+a_{7}\right)\rfloor=\lfloor Nx\rfloor\\
\lfloor N\left(a_{1}+a_{3}+a_{5}+a_{7}\right)\rfloor=\lfloor Nx\rfloor\\
\lfloor N\left(a_{2}+a_{3}+a_{4}+a_{5}\right)\rfloor=\lfloor Nx\rfloor\\
\lfloor N\left(a_{1}+a_{3}+a_{4}+a_{6}\right)\rfloor=\lfloor Nx\rfloor\\
\sum_{i=0}^{7}a_{i}=1
\end{array}\right..\label{eq:constr_y4}
\end{equation}
\end{itemize}
In the limit $N\to\infty$, due to the 7 constraints in eq.~(\ref{eq:constr_y4}),
the summation over elements in the box $\left[0,1\right]^{8}$ in
eq.~(\ref{eq:y4fm}) can be replaced by an integral over the interval
\begin{equation}
\mathcal{B}_{x}\equiv\left[0,\,\min\left\{ \frac{x}{2},1-\frac{3}{2}x\right\} \right]\qquad\text{for \ensuremath{x<\frac{2}{3}}},
\end{equation}

while for $x>\frac{2}{3}$ the constraints admit no solutions and
$\mathbb{E}\left[\mathcal{Z}_{y=4}\left(x,K,\boldsymbol{\xi}\right)\right]\cong0$.

Therefore, for $x<\frac{2}{3}$, we can write 
\begin{align*}
\mathbb{E}\left[\mathcal{Z}_{y=4}\left(x,K,\boldsymbol{\xi}\right)\right] & \cong2^{N}\int_{\mathcal{B}_{x}}db\binom{N}{N\left(1-b-\frac{3}{2}x\right),Nb,Nb,N\left(\frac{x}{2}-b\right),Nb,N\left(\frac{x}{2}-b\right),N\left(\frac{x}{2}-b\right),Nb}f_{1}^{y=4}\left(x,K\right)\\
 & \cong\int_{\mathcal{B}_{x}}db\ e^{N\left(\ln2+H_{8}\left(x,b\right)+\ln f_{1}^{y=4}\left(x,K\right)\right)}\\
 & \cong e^{N\left(\ln2+H_{8}\left(x,b^{*}(x)\right)+\ln f_{1}^{y=4}\left(x\right)\right)},
\end{align*}
where in the last line we estimated the integral with its saddle point
contribution at $b^{\star}\left(x\right)=\argmax_{b\in\mathcal{B}_{x}}H_{8}\left(x,b\right)$.
The function $H_{8}\left(x,b\right)$ is the Shannon entropy of an
eight-states discrete probability distribution with masses given by
the components of the vector $\left(1-b-3/2\ x,b,b,x/2-b,b,x/2-b,x/2-b,b\right)$.
It follows that the first moment upper bound to the storage capacity
for quadruplets of solutions at a fixed distance $x$ is given by
\[
\alpha_{UB}^{y=4}\left(x,K\right)=-\frac{\ln2+H_{8}\left(x,b^{\star}\left(x\right)\right)}{\ln f_{1}^{y=4}\left(x,K\right)}.
\]

The numerical evaluation of the two upper bounds, $y=3$ and $y=4$,
can be found in fig.~\ref{fig:Upper-bounds-for} along with the predictions
for the upper bound from non-rigorous calculations for larger $y$'s.

\subsection{Upper bounds under symmetric assumption for saddle point}

Since a rigorous expression for the upper bound $\alpha_{UB}^{y}\left(x,K\right)$
for $y>4$ is hard to derive, due to highly non-trivial combinatorial
factors, we resort to non-rigorous field theoretical techniques and
replica symmetric ansatz to obtain an expression that we believe to
be exact for low values of $y$ but is likely slightly incorrect for
very large $y$ due to replica symmetry breaking effects. The generic
computation of the $n$-th moment of the partition function, $\mathbb{E}\left[\mathcal{Z}_{y}^{n}\right]$,
is shown in Appendix~\ref{sec:n-th-moment-using}. Here we present
the final result for the first moment bound, i.e. the case $n=1$.

In what follows, we denote with $\mathrm{SP}$ the saddle point operation,
and we use the overlap between solutions $q_{1}\equiv1-2x$ as our
control parameter instead of the distance $x$ to match the usual
notation of replica theory calculations. Up to subleading terms in
$N$ as $N\to\infty$ we have: 

\global\long\def\sp{\underset{\hat{q}_{1}}{\text{\ensuremath{\mathrm{SP}}}}}%

\begin{eqnarray*}
\mathbb{E}\left[\mathcal{Z}_{y}\left(q_{1},K,\boldsymbol{\xi}\right)\right] & \cong & e^{N\left(\sp\left\{ G_{IS}^{n=1,y}\left(q_{1},\hat{q}_{1}\right)\right\} +\alpha G_{E}^{n=1,y,K}\left(q_{1}\right)\right)},
\end{eqnarray*}
where 
\begin{eqnarray*}
G_{IS}^{n=1,y}\left(q_{1},\hat{q}_{1}\right) & = & -q_{1}\hat{q}_{1}\frac{y\left(y-1\right)}{2}-\frac{\hat{q}_{1}y}{2}+\ln\int Dt\left(2\cosh\left(t\sqrt{\hat{q}_{1}}\right)\right)^{y}\\
G_{E}^{n=1,y,K}\left(q_{1}\right) & = & \ln\int Dz\left[\sum_{s=\pm1}s\,H\left(\frac{-s\,K}{\sqrt{1-q_{1}}}+\frac{\sqrt{q_{1}}z}{\sqrt{1-q_{1}}}\right)\right]^{y}
\end{eqnarray*}
where we have used the shorthand notation for standard Gaussian integrals
$Dz\equiv dz\,\frac{e^{-\frac{x^{2}}{2}}}{\sqrt{2\pi}}$, and the
definition $H\left(x\right)=\int_{x}^{\infty}Dz=\frac{1}{2}\mathrm{erfc}\left(\frac{x}{\sqrt{2}}\right)$.

The first moment bound therefore implies that in the limit $N\to\infty$
there are no SAT-$x$ multiplets of $y$ solutions if

\begin{equation}
\sp\left\{ G_{IS}^{n=1,y}\left(q_{1},\hat{q}_{1}\right)\right\} +\alpha G_{E}^{n=1,y,K}\left(q_{1}\right)<0.
\end{equation}

This leads to an estimation $\alpha_{UB,S}^{y}$ given by the symmetric
saddle point of the true upper bound $\alpha_{UB}^{y}$ that takes
the form

\begin{equation}
\alpha_{UB,S}^{y}\left(q_{1},K\right)\equiv-\frac{\sp\left\{ G_{IS}^{n=1,y}\left(q_{1},\hat{q}_{1}\right)\right\} }{G_{E}^{n=1,y,K}\left(q_{1}\right)}.\label{eq:UB_symm}
\end{equation}

These expressions are derived under a symmetric ansatz (i.e. we restrict
the search for the saddle point to a particular subset of the region
of integration) and thus are not rigorous; yet the results in the
cases $y=2,3,4$ agree with the rigorous ones. The corresponding curves
are shown in fig.~\ref{fig:Upper-bounds-for}. Notice that for some
values and $y$ and $x$, the second moment upper bound is larger
than the critical value for the single (and less constrained) system
, $\alpha_{UB}^{y}\left(x\right)>\alpha_{c}$. Since the replicated
system critical value, if exist, is such that $\alpha_{c}^{y}\left(x\right)\leq\alpha_{c}$
, in that parameter region the upper bound is not tight.

\begin{figure}
\begin{centering}
\includegraphics[width=1\columnwidth]{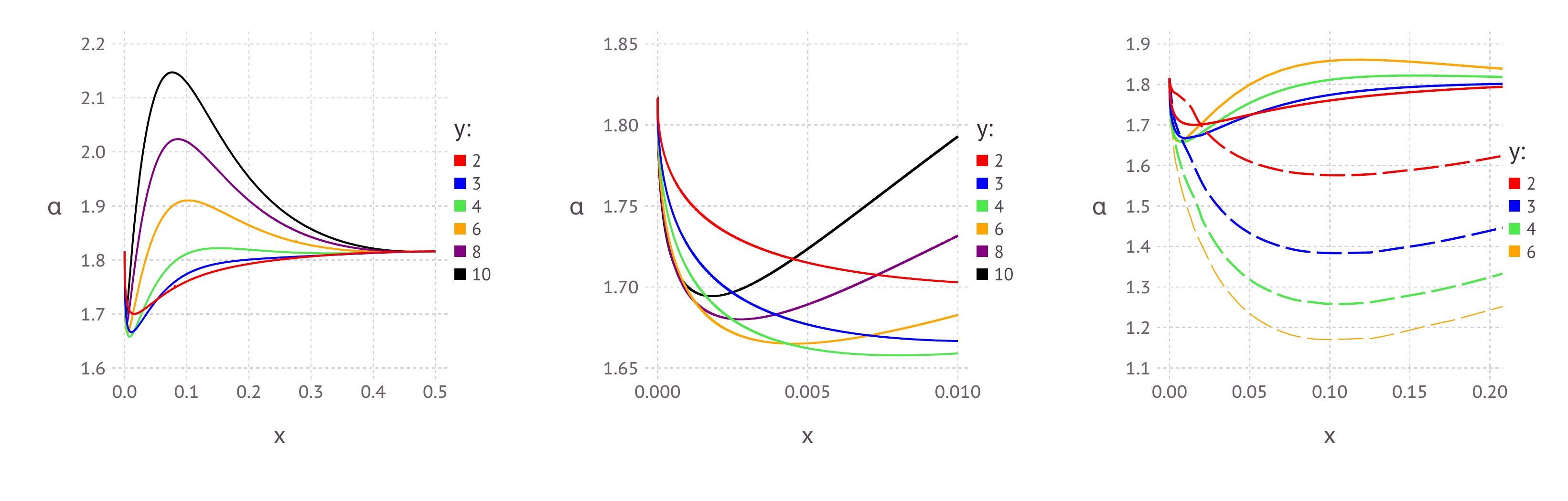}
\par\end{centering}
\caption{\label{fig:Upper-bounds-for}(Left) Upper bound $\alpha_{UB}^{y}\left(x,K=1\right)$
to the SAT/UNSAT threshold for the RBP problem with $y$ replicas
constrained at pairwise distance $x$. Curves are given by rigorous
derivation ($y=2,3,4)$ or by non-rigorous field theoretical calculations
(\ref{eq:UB_symm}) ($y>4$) . (Center) Zoom of the figure on the
left. Close to $x=0$ the curves corresponding to different $y$ intersect.
(Right) The upper bounds (solid lines) are compared to the $S$ point
predictions (\ref{eq:LB_symm}) for the lower bounds (dashed lines). }
\end{figure}

As one can see, the curves intersect in a nontrivial way. Let's take
for example the curves for $y=2$ and $y=3$. If the bounds were tight
for all values of $x$, the curve at $y=3$ should always stay below
the curve for $y=2$. This follows directly from the fact that if
we have no way of accommodating pairs of solutions then we do not
have a way to accommodate triplets solutions either. Instead, the
fact that the curves intersect means that for values of $x$ smaller
than the intersection point the bounds stop being tight. This straightforward
argument, generalized to higher values of $y$, therefore we can define
a tighter upper bound, that we call $\tilde{\alpha}_{UB}^{y}\left(x,K\right),$
for the existence of sets of $y$ constrained solutions:

\begin{equation}
\tilde{\alpha}_{UB}^{y}\left(x,K\right)=\min\left\{ \alpha_{UB}^{y^{\prime}}\left(x,K\right)\,:\,y^{\prime}\in\mathbb{N},\,2\leq y^{\prime}\leq y\right\} .\label{eq:supp}
\end{equation}

\subsection{Lower bounds under symmetric assumption for the saddle point}

We compute the partition function moments needed for the lower bounds
in Appendix~\ref{sec:n-th-moment-using}. The final result of the
replica calculation is given by

\global\long\def\spp{\underset{q_{0}\ \hat{q}_{0}\ \hat{q}_{1}}{\text{\ensuremath{\mathrm{SP}}}}}%

\global\long\def\sqz{\underset{q_{0}}{\text{\ensuremath{\mathrm{SP}}}}}%

\begin{eqnarray}
\mathbb{E}\left[\mathcal{Z}_{y}^{2}\left(q_{1},K,\boldsymbol{\xi}\right)\right] & \cong & \exp\left(N\cdot\spp\Biggl\{-\hat{q}_{1}y-y\left[yq_{0}\hat{q}_{0}+\left(y-1\right)q_{1}\hat{q}_{1}\right]\right.\nonumber \\
 &  & +\ln\int Dz\left[\int Dt\left(2\cosh\left(\sqrt{\hat{q}_{0}}z+\sqrt{\hat{q}_{1}-\hat{q}_{0}}t\right)\right)^{y}\right]^{2}\nonumber \\
 &  & +\ \alpha\ln\int Dz\left[\int Dt\left[\sum_{s=\pm1}s\,H\left(\frac{-s\,K}{\sqrt{1-q_{1}}}+\frac{\sqrt{q_{0}}z+\sqrt{q_{1}-q_{0}}t}{\sqrt{1-q_{1}}}\right)\right]^{y}\right]^{2}\\
 & = & \exp\left(N\cdot\spp\left\{ G_{I}^{n=2,y}\left(q_{0},\hat{q}_{0},q_{1},\hat{q}_{1}\right)+G_{S}^{n=2,y}\left(\hat{q}_{0},\hat{q}_{1}\right)+\alpha G_{E}^{n=2,y,K}\left(q_{0},q_{1}\right)\right\} \right),
\end{eqnarray}
where

\begin{eqnarray}
G_{I}^{n=2,y}\left(q_{0},\hat{q}_{0},q_{1},\hat{q}_{1}\right) & = & -\hat{q}_{1}y-y\left[yq_{0}\hat{q}_{0}+\left(y-1\right)q_{1}\hat{q}_{1}\right]\\
G_{S}^{n=2,y}\left(\hat{q}_{0},\hat{q}_{1}\right) & = & \ln\int Dz\left[\int Dt\left(2\cosh\left(\sqrt{\hat{q}_{0}}z+\sqrt{\hat{q}_{1}-\hat{q}_{0}}t\right)\right)^{y}\right]^{2}\\
G_{E}^{n=2,y,K}\left(q_{0},q_{1}\right) & = & \ln\int Dz\left[\int Dt\left[\sum_{s=\pm1}s\,H\left(\frac{-s\,K}{\sqrt{1-q_{1}}}+\frac{\sqrt{q_{0}}z+\sqrt{q_{1}-q_{0}}t}{\sqrt{1-q_{1}}}\right)\right]^{y}\right]^{2}.\label{eq:RSAn2y}
\end{eqnarray}
Performing the saddle points over the variables $\hat{q}_{0}$ and
$\hat{q}_{1}$, these expressions reduce to 
\begin{eqnarray}
\mathbb{E}\left[\mathcal{Z}_{y}^{2}\left(q_{1},K,\boldsymbol{\xi}\right)\right] & \simeq & e^{N\left(\max_{q_{0}}\left\{ G_{IS}^{\mathrm{opt},n=2,y}\left(q_{0},q_{1}\right)+\alpha G_{E}^{n=2,y,K}\left(q_{0},q_{1}\right)\right\} \right)},\label{eq:q00}
\end{eqnarray}
where
\begin{equation}
G_{IS}^{\mathrm{opt},n=2,y}\left(q_{0},q_{1}\right)=\underset{\hat{q}_{0}\ \hat{q}_{1}}{\text{\ensuremath{\mathrm{SP}}}}\left\{ G_{I}^{n=2,y}\left(q_{0},\hat{q}_{0},q_{1},\hat{q}_{1}\right)+G_{S}^{n=2,y}\left(\hat{q}_{0},\hat{q}_{1}\right)\right\} .
\end{equation}

For fixed $\alpha$, if the optimum in eq.~(\ref{eq:q00}) is at
$q_{0}=0$ we have $\mathbb{E}\left[\mathcal{Z}_{y}^{2}\left(q_{1},K,\boldsymbol{\xi}\right)\right]\cong\mathbb{E}\left[\mathcal{Z}_{y}\left(q_{1},K,\boldsymbol{\xi}\right)\right]^{2}$
and from the second moment inequality, eq.~(\ref{eq:secmom}), we
have that there is positive probability of finding multiplets of $y$
solutions at distance $x=\frac{1}{2}\left(1-q_{1}\right).$ This in
turn implies that the lower bound is valid for all $\alpha$'s such
that $\argmax_{q_{0}}\left\{ G_{IS}^{opt,n=2,y}\left(q_{0},q_{1}\right)+\alpha G_{E}^{n=2,y,K}\left(q_{0},q_{1}\right)\right\} =0$.
In particular the symmetric saddle point prediction for the lower
bound is given by

\begin{equation}
\alpha_{LB,S}^{y}\left(q_{1}\right)=\sup\left\{ \alpha\geq0\,\bigg|\,\argmax_{q_{0}}\left\{ G_{IS}^{opt,n=2,y}\left(q_{0},q_{1}\right)+\alpha G_{E}^{n=2,y,K}\left(q_{0},q_{1}\right)\right\} =0\right\} .\label{eq:LB_symm}
\end{equation}
The results for $y=2,3,4,5$ are summarized in fig.~\ref{fig:Upper-bounds-for}
on the right. In fig.~\ref{fig:Bounds-for-different} we plot an
enlargement of the small-distance region around $x=0$, corresponding
to $q_{1}=1$. We find that in all cases there is an inconsistency
region $\text{\ensuremath{\left[0,x_{c}\left(y\right)\right]}}$ in
which the symmetric lower and upper bounds switch roles, similarly
to what happened in the case of $y=2$ (see fig.~\ref{fig:fig2UBLB}).
The true lower bound cannot thus be symmetric in this region: the
configuration in which the two SAT-$x$ multiplets of $y$ solutions
are collapsed on a single multiplet always gives a better saddle point,
resulting in a lower bound equal to the upper bound. We thus conjecture
that for $x<x_{c}\left(y\right)$ the bounds are tight, like in the
$y=2$ case. The symmetry of lower and upper bounds with respect to
$x$ on the interval $\left[0,1\right]$ (or the corresponding symmetry
for $q_{1}$) which holds for $y=2$ does not apply to general $y$.
In our numerical exploration presented in fig.~\ref{fig:Upper-bounds-for},
we focused on the region of small $x$. We also notice that the lower
bounds for increasing $y$'s decrease monotonically, and in the limit
$y\to\infty$ the limiting curve seem to exhibit a vertical asymptote
for $x=0$. Furthermore, the intersection point $x_{c}\left(y\right)$
seem to decrease monotonically with $y$ and to approach zero. It
is also worth noting that, for all the $y$ that we tested, we found
that in the region $\left[0,x_{c}\text{\ensuremath{\left(y\right)}}\right]$
we have $\tilde{\alpha}_{UB}^{y}=\alpha_{UB}^{y}$, which is consistent
with the conjecture that the bounds are tight in this region.

\begin{figure}[h]
\centering{}\includegraphics[width=1\columnwidth]{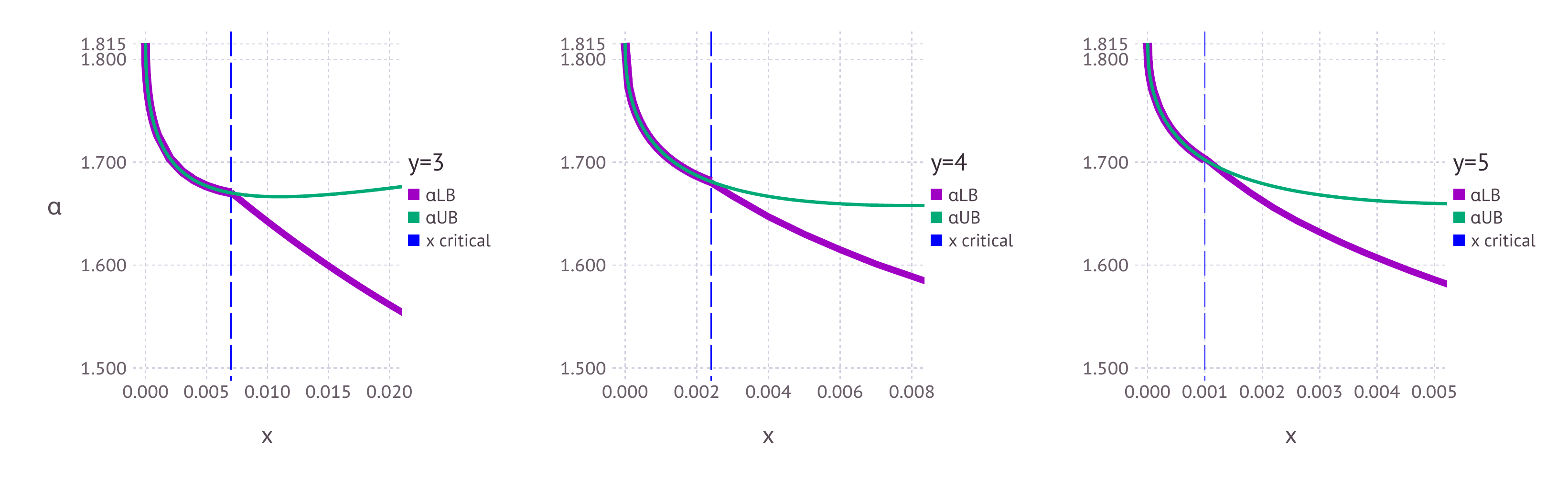}\caption{\label{fig:Bounds-for-different}Lower and upper bounds for the RBP
with $K=1$ and for different values of $y=3,4,5$, in the region
of small $x$. Like in the case of $y=2$, for $x$ larger than the
critical value $x_{c}\left(y\right)$ (blue vertical line) there is
a gap between the symmetric lower bound (purple line) and the upper
bound (green line). This gap closes in correspondence of the $SB$
solution for $x\protect\leq x_{c}\left(y\right)$ and the two bounds
coincide.}
\end{figure}

\section{Conclusions}

We have presented an investigation of the geometry of the solutions
space for the binary symmetric perceptron model storing random patterns.
According to the non-rigorous analysis conducted with the replica
method, this model exhibits the same qualitative phenomenology as
the more standard non-symmetric counterpart. In particular, we focused
on signatures for the presence of rare dense regions of solutions,
which are of particular interest since according to previous studies
they appear to be crucially connected to the existence of efficient
learning algorithms \citep{baldassi2015subdominant,baldassi2016unreasonable}.
The analogous structures for continuous models (of the kind used for
deep learning applications) are wide flat minima, which have also
been related to training efficiency and generalization capabilities
\citep{baldassi2020shaping}.

Compared to standard models, the symmetry in the model used for this
paper simplifies the analytical treatment, as was first shown in ref.~\citep{aubin2019storage}.
Thanks to this, we have been able to show rigorously (up to a numerical
optimization step) that in the large $N$ limit there exist an exponential
number of pairs of solution at arbitrary $O\left(N\right)$ Hamming
distance. A further analysis led us to conjecture that this scenario
extends to multiplets of more than 2 solutions at fixed distance.
These results are highly non-trivial, and consistent with the replica
analysis; a complete and rigorous confirmation will presumably require
different tools or alternative approaches however, and thus remains
as an open problem. Besides this, several other important problems
related to the dense regions, with potentially far-fetching practical
and theoretical implications, remain open: in particular, obtaining
a detailed description of their geometry, and a complete characterization
of their accessibility by efficient algorithms.

\bibliographystyle{unsrturl}
\phantomsection\addcontentsline{toc}{section}{\refname}\bibliography{bibliography}

\appendix

\section{$y\to\infty$ limit\label{subsec:Limit}}

\global\long\def\spmi{\underset{\delta q\hat{q}_{1}\delta\hat{q}}{\text{\ensuremath{\mathrm{SP}}}}}%

In this Section we derive the large $y$ limit of the entropy

\[
\phi_{y}\left(x,K,\alpha\right)=\lim_{N\to\infty}\frac{1}{yN}\mathbb{E}_{\boldsymbol{\xi}}\ln\mathcal{Z}_{y}\left(x,K,\boldsymbol{\xi}\right)
\]

within RS assumptions. For convenience of notation we will use the
overlap $q_{1}=1-2x$ instead of $x$. As explained in ref.~\citep{baldassi2020shaping},
the computation of $\phi_{y}\left(x\right)$ is formally equivalent
to that of a single replica in the 1RSB ansatz with Parisi parameter
$y$, except for the fact that $q_{1}$ is fixed externally instead
of being optimized as usual. We obtain the following entropy for the
Random Binary Perceptron (RBP) with $y$ real replicas:

\begin{align}
\phi_{y}\left(x,K,\alpha\right) & =\spp\left\{ -\frac{\hat{q}_{1}}{2}\left(1-q_{1}\right)+\frac{y}{2}\left(q_{0}\hat{q}_{0}-q_{1}\hat{q}_{1}\right)+\frac{1}{y}\int Dz_{0}\ln\int Dz_{1}\ \left[2\cosh\left(\sqrt{\hat{q}_{0}}z_{0}+\sqrt{\hat{q}_{1}-\hat{q}_{0}}z_{1}\right)\right]^{y}+\right.\\
+\frac{\alpha}{y} & \left.\int Dz_{0}\ \ln\int Dz_{1}\left[\sum_{s=\pm1}s\,H\left(\frac{-s\,K}{\sqrt{1-q_{1}}}+\frac{\sqrt{q_{0}}z_{0}+\sqrt{q_{1}-q_{0}}z_{1}}{\sqrt{1-q_{1}}}\right)\right]^{y}\right\} .
\end{align}
We want to take the limit $y\to\infty$ in the previous expression.
By looking at the entropic and energetic parts we derive the appropriate
scalings

\begin{equation}
\hat{q}_{0}=\hat{q}_{1}-\frac{\delta\hat{q}}{y},\quad q_{0}=q_{1}-\frac{\delta q}{y},
\end{equation}
and the previous equation becomes
\begin{equation}
\phi_{y=\infty}\left(q_{1},K,\alpha\right)=\spmi\left\{ -\frac{\hat{q}_{1}}{2}\left(1-q_{1}\right)-\frac{1}{2}\left(\delta q\hat{q}_{1}+\delta\hat{q}q_{1}\right)+\int Dz_{0}\ A^{\star}\left(z_{0}\right)+\alpha\int Dz_{0}\ B^{\star}\left(z_{0}\right)\right\} ,
\end{equation}
where

\begin{align}
A^{\star}\left(z_{0}\right) & =\ln2-\min_{z_{1}}\left\{ \frac{z_{1}^{2}}{2}-\ln\cosh\left(\sqrt{\hat{q}_{1}}z_{0}+\sqrt{\delta\hat{q}}z_{1}\right)\right\} ,\\
B^{\star}\left(z_{0}\right) & =-\min_{z_{1}}\left\{ \frac{z_{1}^{2}}{2}-\ln\left[\sum_{s=\pm1}s\,H\left(\frac{-s\,K}{\sqrt{1-q_{1}}}+\frac{\sqrt{q_{1}}z_{0}+\sqrt{\delta q}z_{1}}{\sqrt{1-q_{1}}}\right)\right]\right\} .
\end{align}
The results are shown in fig.~\ref{fig:clusterone}. The behavior
of $\phi_{y=\infty}\left(q_{1}\right)$ close to $q_{1}=1$, where
it approaches the maximum volume curve, reveals the existence of a
dense cluster of solutions. Furthermore, the maximum volume curve
coincides with the curve for $\alpha=0$, which means that there are
no constraints to impose and the function $\phi_{y=\infty}\left(q_{1},K,0\right)=H_{2}\left(\left(1+\sqrt{q_{1}}\right)/2\right).$
We expect the value obtained within the RS ansatz for $\phi_{y=\infty}\left(q_{1},K,\alpha\right)$
to not be the correct one, at least for $\alpha$ above some critical
value where spin glass instabilities arise. In fact $\phi_{y=\infty}\left(q_{1},K,\alpha\right)$
yields a SAT/UNSAT transition that is wrong, since it is above the
known one for the standard $y=1$ model. Therefore this scenario should
be checked within a 1RSB calculation, where we also expect the dense
cluster prediction to remain true. We refer to \citep{baldassi2015subdominant}
for an in-depth analysis of a similar model which takes also into
account replica symmetry breaking corrections. 

\section{Derivation of the lower bound}

\subsection{Change of integration variables in second moment bound\label{subsec:Change-of-integration}}

The bound in eq.~(\ref{eq:boundsm}) depends on the 8 variables $\mathbf{a}$.
We want now to reduce the number of from 8 to 5 using the constraints
in eq\@.~(\ref{eq:const}). We choose to write $a_{0}$, $a_{6}$,
$a_{7}$ as functions of the other variables

\begin{equation}
\left\{ \begin{array}{l}
a_{0}=1-a_{1}-a_{2}-a_{3}-x\\
a_{6}=x-a_{1}-a_{2}-a_{5}\\
a_{7}=a_{1}+a_{2}-a_{4}
\end{array}\right..\label{eq:rep1}
\end{equation}
The integration set $V_{x}$ is then reparametrized as a function
of the variables $\vec{a}:=\left(a_{1},a_{2},a_{3},a_{4},a_{5}\right)$.
We indicate with $\mathbf{a}\left(\vec{a},x\right)$ the immersion
from $\mathbb{R}^{5}$ to $\mathbb{R}^{8}$ whose components from
$a_{1}$ to $a_{5}$ are mapped in themselves while the remaining
ones are specified by the equations in~(\ref{eq:rep1}). This makes
the expression of $V_{x}$ more explicit and lets us rewrite the integral
in an equivalent way. The integration set becomes $V_{x}^{\prime}\subseteq\left[0,1\right]^{5}$
and it is specified by the following set of inequalities:

\begin{equation}
\left\{ \begin{array}{l}
0\leq a_{i}\leq1\;\forall\ i=1,\dotsc,5\\
0\leq a_{1}+a_{2}-a_{4}\leq1\\
a_{1}+a_{2}+a_{5}\leq x\\
a_{1}+a_{2}+a_{3}\leq1-x
\end{array}\right..\label{eq:Vprime}
\end{equation}
With this change of variables eq.~(\ref{eq:boundsm}) becomes:

\begin{equation}
\mathbb{E}\left[\mathcal{Z}_{y=2}^{2}\left(x,K,\boldsymbol{\xi}\right)\right]\leq C_{0}N^{3/2}\int_{V_{x}^{\prime}}d\vec{a}\ e^{N\left[\ln2+H_{8}\left(\vec{a},x\right)+\alpha\ln f_{2}\left(\vec{a},x,K\right)\right]},\label{eq:boundsm2}
\end{equation}
where we defined $H_{8}\left(\vec{a},x\right):=H_{8}\left(\mathbf{a}\left(\vec{a},x\right)\right)$
and $f_{2}\left(\vec{a},x,K\right):=f_{2}\left(\mathbf{a}\left(\vec{a},x\right),x,K\right)$.
The covariance matrix in the Gaussian integral $f_{2}\left(\vec{a},x,K\right)$
is reparameterized in the following way (cf.~eq.~(\ref{eq:alloverlaps})):

\begin{equation}
\Sigma=\left(\begin{array}{cccc}
1 & q_{1} & q_{01} & q_{02}\\
q_{1} & 1 & q_{03} & q_{04}\\
q_{01} & q_{03} & 1 & q_{1}\\
q_{02} & q_{04} & q_{1} & 1
\end{array}\right)\quad\textrm{where}\quad\left\{ \begin{array}{l}
q_{1}=1-2x\\
q_{01}=1-2\left(x+a_{2}+a_{3}-a_{4}-a_{5}\right)\\
q_{02}=1-2\left(2a_{1}+a_{2}+a_{3}-a_{4}+a_{5}\right)\\
q_{03}=1-2\left(a_{2}+a_{3}+a_{4}+a_{5}\right)\\
q_{04}=1-2\left(x-a_{2}+a_{3}+a_{4}-a_{5}\right)
\end{array}\right..\label{eq:syseq}
\end{equation}

The next and final reparametrization of the integral is suggested
by the form of the covariance matrix. In particular we would like
to express the four possible overlaps between the two pairs of solution
using the four parameters $q_{01}$, $q_{02}$, $q_{03}$, $q_{04}$
and group them in a four dimensional vector $\vec{q}_{0}$. However,
since our integration domain is 5-dimensional, we need an additional
parameter that we call $\eta$. Inverting the under-parametrized system
of eqs.~(\ref{eq:syseq}), we obtain the vectors $\vec{a}^{\star}$
that lie in the vector space below, for $\eta\in\mathbb{R}$:

\begin{equation}
\left\{ \begin{array}{l}
a_{1}^{\star}=\frac{1}{4}\left(q_{01}-q_{02}+2x\right)-\eta\\
a_{2}^{\star}=\frac{1}{4}\left(-q_{03}+q_{04}+2x\right)-\eta\\
a_{3}^{\star}=\frac{1}{4}\left(2-q_{01}-q_{04}-4x\right)+\eta\\
a_{4}^{\star}=\frac{1}{4}\left(q_{01}-q_{03}+2x\right)-\eta\\
a_{5}^{\star}=\eta
\end{array}\right..\label{eq:syseqris}
\end{equation}
By constraining the solutions $\vec{a}^{\star}$ in their natural
domain $V_{x}^{\prime}$ we find how the domain is transformed in
the new coordinates $\vec{q}_{0}$ and $\eta$: 

\begin{equation}
\left\{ \begin{array}{l}
\frac{1}{4}\left(q_{01}-q_{02}+2x-4\right)\leq\eta\leq\frac{1}{4}\left(q_{01}-q_{02}+2x\right)\\
\frac{1}{4}\left(-q_{03}+q_{04}+2x-4\right)\leq\eta\leq\frac{1}{4}\left(-q_{03}+q_{04}+2x\right)\\
\frac{1}{4}\left(q_{01}+q_{04}+4x-2\right)\leq\eta\leq\frac{1}{4}\left(q_{01}+q_{04}+4x+2\right)\\
\frac{1}{4}\left(q_{01}-q_{03}+2x-4\right)\leq\eta\leq\frac{1}{4}\left(q_{01}-q_{03}+2x\right)\\
0\leq\eta\leq1\\
\frac{1}{4}\left(q_{01}-q_{02}-q_{03}+q_{04}\right)\leq\eta\\
\frac{1}{4}\left(-q_{02}+q_{04}+2x-4\right)\leq\eta\leq\frac{1}{4}\left(-q_{02}+q_{04}+2x\right)\\
\frac{1}{4}\left(-q_{02}-q_{03}+4x-2\right)\leq\eta
\end{array}\right.,\label{eq:syseqris2-1}
\end{equation}
where we have expressed all inequalities in terms of the variable
$\eta$. This set of inequalities specifies a new integration domain
in eq.~(\ref{eq:boundsm2}), this time in the new variables $\eta$
and $\vec{q}_{0}$, that we call $\tilde{V}_{x}$ and that depends
on $x$. Again, we can express the vector of solutions $\vec{a}^{\star}$
as a function of the pair $\left(\vec{q}_{0},\eta\right)$. The integral~(\ref{eq:boundsm2})
is rewritten as:

\begin{equation}
\mathbb{E}\left[\mathcal{Z}_{y=2}^{2}\left(x,K,\xi\right)\right]\leq C_{0}N^{3/2}\int_{\tilde{V}_{x}}d\vec{q}_{0}\ d\eta\ e^{N\left[\ln2+H_{8}\left(\vec{q}_{0},\eta,x\right)+\alpha\ln f_{2}\left(\vec{q}_{0},x,K\right)\right]},\label{eq:boundsm3-1}
\end{equation}
where we adopt the convention that $f_{2}\left(\vec{q}_{0},x,K\right):=f_{2}\left(\vec{a}^{\star}\left(\vec{q}_{0},\eta\right),x,K\right)$
and $H_{8}\left(\vec{q}_{0},\eta,x\right):=H_{8}\left(\vec{a}^{\star}\left(\vec{q}_{0},\eta\right),x\right)$. 

\subsection{\label{subsec:A-Proof-of-Lemma-2}Proof of Lemma \ref{Lemma02}}
\begin{proof}[Proof of Lemma \ref{Lemma02}]
From eq.~(\ref{eq:eqsympN}) we obtain the following inequalities:

\begin{equation}
\left\{ \begin{array}{l}
\vert-a_{0}+1-a_{1}-a_{2}-a_{3}-x\vert<\frac{3}{N}\\
\vert a_{6}-x+a_{1}+a_{2}+a_{5}\vert<\frac{1}{N}\\
\vert a_{7}-a_{1}-a_{2}+a_{4}\vert<\frac{2}{N}
\end{array}\right..\label{eq:finiteSum}
\end{equation}
In the limit $N\to\infty$ these inequalities determine three of the
parameters as a function of the other five:

\begin{equation}
\left\{ \begin{array}{l}
a_{0}^{\star}=1-a_{1}-a_{2}-a_{3}-x\\
a_{6}^{\star}=x-a_{1}-a_{2}-a_{5}\\
a_{7}^{\star}=a_{1}+a_{2}-a_{4}
\end{array}\right..\label{eq:substitute}
\end{equation}
Notice that the summation on the left hand side of eq.~(\ref{eq:discont})
is taken for $\textbf{a}\in\left\{ 0,1/N,2/N,\ldots,1\right\} ^{8}$.
If we fix the five components vector $\vec{a}:=\left(a_{1},\dotsc,a_{5}\right)\in V{}_{x}^{\prime}\cap\left\{ 0,1/N,2/N,\ldots,1\right\} ^{5}$
where $V_{x}^{\prime}$ is defined as in eq.~(\ref{eq:Vprime}),
then, independently from this 5-dimensional vector, there exist at
most a fixed number of $\textbf{a}$'s that satisfy the inequalities
in eq.~(\ref{eq:finiteSum}) (for every $N$ and $x\in\left[0,1\right]$).
This is sufficient to conclude that for large enough $N$ there exists
a positive constant $F_{0}$ such that

\begin{align*}
\sum_{\mathbf{a}\in V_{N,x}\cap\left\{ 0,1/N,2/N,\ldots,1\right\} ^{8}}\binom{N}{Na_{0}\,\dotsc\,Na_{7}}\ \psi\left(\mathbf{a}\right)^{N} & \leq F_{0}\sum_{\vec{a}\in V_{x}^{\prime}\cap\left\{ 0,1/N,2/N,\ldots,1\right\} ^{5}}\binom{N}{\left\lfloor Na_{0}^{\star}\right\rfloor \,Na_{1}\,\dotsc\,Na_{5}\,\left\lfloor Na_{6}^{\star}\right\rfloor \,Na_{7}^{\star}}\ \times\\
 & \qquad\times\ \psi\left(a_{0}^{\star},a_{1},\dotsc,a_{5},a_{6}^{\star},a_{7}^{\star}\right)^{N}.
\end{align*}
where $V_{x}^{\prime}$ is defined by the system of eqs.~(\ref{eq:Vprime}).

From Stirling's approximation, the expression for large $N$ and fixed
$a_{i}$ of the multinomial factor is

\begin{align*}
\binom{N}{Na_{0}\,\dotsc\,Na_{m}} & =e^{NH\left(\mathbf{a}\right)-\frac{m-1}{2}\ln N+\mathcal{O}(1)}\\
 & \leq G_{0}e^{NH\left(\mathbf{a}\right)-\frac{m-1}{2}\ln N}
\end{align*}
where $G_{0}$ is some positive constant and $H\left(\mathbf{a}\right)$
is the Shannon entropy of the discrete probability distribution with
masses $\left\{ a_{0},\dotsc,a_{m}\right\} .$ Putting all together
we have
\begin{align*}
 & \sum_{\mathbf{a}\in V_{N,x}\cap\left\{ 0,1/N,2/N,\ldots,1\right\} ^{8}}\binom{N}{Na_{0}\,\dotsc\,Na_{7}}\ \psi\left(\mathbf{a}\right)^{N}\\
 & \qquad\leq F_{0}\sum_{\vec{a}\in V{}_{x}^{\prime}\cap\left\{ 0,1/N,2/N,\ldots,1\right\} ^{5}}\binom{N}{\left\lfloor Na_{0}^{\star}\right\rfloor \,Na_{1}\,\dotsc\,Na_{5}\,\left\lfloor Na_{6}^{\star}\right\rfloor \,Na_{7}^{\star}}\psi\left(a_{0}^{\star},a_{1},\dotsc,a_{6}^{\star},a_{7}^{\star}\right)^{N}\\
 & \qquad\leq\frac{F_{0}G_{0}}{N^{\frac{7}{2}}}\sum_{\vec{a}\in V{}_{x}^{\prime}\cap\left\{ 0,1/N,2/N,\ldots,1\right\} ^{5}}e^{NH_{8}\left(a_{0}^{\star},a_{1},\dotsc,a_{5},a_{6}^{\star},a_{7}^{\star}\right)-\frac{7}{2}\ln N}\psi\left(a_{0}^{\star},a_{1},\dotsc,a_{5},a_{6}^{\star},a_{7}^{\star}\right)^{N}\\
 & \qquad<C_{0}N^{\frac{3}{2}}\int_{V_{x}}d\mathbf{a}\ e^{N\left[H_{8}\left(\mathbf{a}\right)+\ln\psi\left(\mathbf{a}\right)\right]}
\end{align*}
where we have used the limit of Riemann sums in the last step and
$C_{0}>F_{0}G_{0}$ is a positive constant that does not depend on
$N$ but depends on $x$. The integral in the last line is defined
as in the footnote for Lemma~\ref{Lemma02}.
\end{proof}

\subsection{\label{subsec:Proof-of-eq}Proof of eq.~(\ref{eq:ineqAlpha0})}

For finite $N$ we define $\mathcal{N}_{2}\left(x\right)$ and $\mathcal{N}_{4}\left(x,\mathbf{a}\right)$
as follows. First, 
\[
\mathcal{N}_{2}\left(x\right)\equiv\sum_{\left\{ \mathbf{w}^{1}\right\} }\sum_{\left\{ \mathbf{w}^{2}\right\} }\ind\left(d_{H}\left(\mathbf{w}^{1},\mathbf{w}^{2}\right)=\left\lfloor Nx\right\rfloor \right),
\]
 which implies that
\begin{align*}
\left(\mathcal{N}_{2}\left(x\right)\right)^{2} & =\left(\sum_{\left\{ \mathbf{w}^{1}\right\} }\sum_{\left\{ \mathbf{w}^{2}\right\} }\ind\left(d_{H}\left(\mathbf{w}^{1},\mathbf{w}^{2}\right)=\left\lfloor Nx\right\rfloor \right)\right)^{2}\\
 & =\sum_{\left\{ \mathbf{w}^{1}\right\} }\sum_{\left\{ \mathbf{w}^{2}\right\} }\sum_{\left\{ \mathbf{\tilde{w}}^{1}\right\} }\sum_{\left\{ \mathbf{\tilde{w}}^{2}\right\} }\ind\left(d_{H}\left(\mathbf{w}^{1},\mathbf{w}^{2}\right)=\left\lfloor Nx\right\rfloor \right)\ \ind\left(d_{H}\left(\mathbf{\tilde{w}}^{1},\mathbf{\tilde{w}}^{2}\right)=\left\lfloor Nx\right\rfloor \right).
\end{align*}
Then, for $\mathbf{a}\in V_{N,x}$ we have:
\begin{align*}
\mathcal{N}_{4}\left(x,\mathbf{a}\right) & \equiv\sum_{\left\{ \mathbf{w}^{1}\right\} }\sum_{\left\{ \mathbf{w}^{2}\right\} }\sum_{\left\{ \mathbf{\tilde{w}}^{1}\right\} }\sum_{\left\{ \mathbf{\tilde{w}}^{2}\right\} }\ind\left(d_{H}\left(\mathbf{w}^{1},\mathbf{w}^{2}\right)=\left\lfloor Nx\right\rfloor \right)\ \ind\left(d_{H}\left(\mathbf{\tilde{w}}^{1},\mathbf{\tilde{w}}^{2}\right)=\left\lfloor Nx\right\rfloor \right)\times\\
 & \times\ind\left(d_{H}\left(\mathbf{w}^{1},\mathbf{\tilde{w}}^{1}\right)=\left\lfloor N\left(a_{2}+a_{3}+a_{6}+a_{7}\right)\right\rfloor \right)\ \ind\left(d_{H}\left(\mathbf{w}^{1},\mathbf{\tilde{w}}^{2}\right)=\left\lfloor N\left(a_{1}+a_{3}+a_{5}+a_{7}\right)\right\rfloor \right)\\
 & \times\ind\left(d_{H}\left(\mathbf{w}^{2},\mathbf{\tilde{w}}^{1}\right)=\left\lfloor N\left(a_{2}+a_{3}+a_{4}+a_{5}\right)\right\rfloor \right)\ \ind\left(d_{H}\left(\mathbf{w}^{2},\mathbf{\tilde{w}}^{2}\right)=\left\lfloor N\left(a_{1}+a_{3}+a_{4}+a_{6}\right)\right\rfloor \right).
\end{align*}
From the definitions it follows that $\mathcal{N}_{4}\left(x,\mathbf{a}\right)\leq\left(\mathcal{N}_{2}\left(x\right)\right)^{2}$
and computing the summations gives
\[
2^{N}\frac{N!}{\prod_{i=0}^{7}\left(Na_{i}\right)!}\leq\left(2^{N}\binom{N}{\left\lfloor Nx\right\rfloor }\right)^{2},\quad\forall\boldsymbol{a}\in V_{N,x}\cap\left\{ 0,\frac{1}{N},\dotsc,1\right\} ^{8}.
\]
Taking the logarithm on both sides, dividing by $N$ and taking the
limit for $N\to\infty$, gives the following inequality 
\[
\ln2+H_{8}\left(\mathbf{a}\right)\leq2\log2+2H_{2}(x)\quad\forall\mathbf{a}\in V_{N,x}.
\]
If we apply now the same change of variable of Appendix~\ref{subsec:Change-of-integration}
the result is 
\[
H_{8}\left(\vec{q}_{0},\eta,x\right)\leq\ln2+2H_{2}\left(x\right)\quad\forall\left(\vec{q}_{0},\eta\right)\in\tilde{V}_{x}.
\]

\subsection{Numerical optimization\label{sec:Numerical-optimization}}

We performed the optimization in expression~(\ref{eq:aLB2}) numerically.
We empirically find the objective function to be ridden by many local
minima, therefore we implemented 3 different strategies to partition
the search space and obtain a numerical estimate of the global one.

A first strategy consists in constructing a $4$-dimensional uniformly-spaced
grid for the values of $\vec{q}_{0}$, and then performing Gradient
Descent (GD) starting from these points and selecting the overall
minimum obtained. The downside of this approach is that the the optimization
is very time-consuming. We simulated grids with up to $m=100^{4}$
number of points. We restrict the experiment to the region of small
$x$, in particular $x<x_{c}^{\prime}$. The results are shown in
fig.~\ref{fig:Upper-and-lower4d}. While for $x>x_{c}$, and already
for a low numbers of points $m$, the numerical estimate coincides
with the symmetric point prediction, for $x<x_{c}$ instead, where
we predict the broken symmetry point to yield the true value of $\alpha_{LB}$,
only with the two finest grid spacing we are able to get close to
the theoretical prediction. Overall, the results for this numerical
experiment are in good agreement with theoretical value predicted
for the saddle point by symmetry arguments, supporting our conclusion
that for \textbf{$x<x_{c}$ }lower and upper bounds coincide.

\begin{figure}
\begin{centering}
\includegraphics[width=0.6\columnwidth]{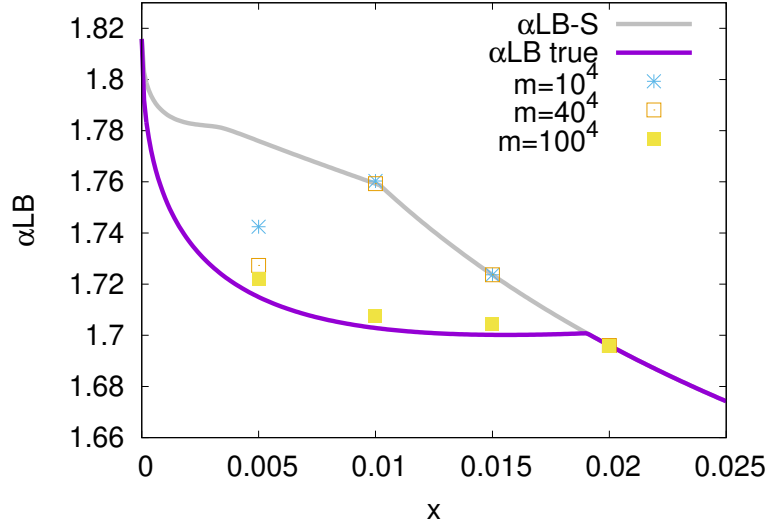}
\par\end{centering}
\caption{\label{fig:Upper-and-lower4d}Numerical lower bounds $\alpha_{LB,y=2}\left(x,K=1\right)$
obtained by multiple restarts of GD from a 4d grids with $m$ points,
for different values of $m$, along with theoretical predictions from
the symmetric point $S$ (that we know to be wrong for $x<x_{c}$)
and the true lower bound (point $S$ for $x>x_{c}$, point $SB$ for
$x<x_{c}$). }
\end{figure}

Another approach is to restrict the search space to a lower dimensional
manifold, containing both the symmetric ($S$) and the symmetry broken
($SB$) points. The lower dimensionality (2 instead of 4) allows us
to use as starting points of our GD procedure grids with smaller spacings.
Therefore, we restrict the search space to points of the type $\vec{q}_{0}=\left(q_{a},q_{b},q_{b},q_{a}\right)$.
The corresponding covariance matrix in this case is given by

\begin{equation}
\Sigma_{SB}=\left(\begin{array}{cccc}
1 & q_{1} & q_{a} & q_{b}\\
q_{1} & 1 & q_{b} & q_{a}\\
q_{a} & q_{b} & 1 & q_{1}\\
q_{b} & q_{a} & q_{1} & 1
\end{array}\right).\label{eq:RSBv}
\end{equation}
The optimization over this submanifold is done by multiple restarts
of GD from a 2-dimensional grid corresponding of values for $q_{a}$
and $q_{b}$. The results are reported in fig.~\ref{fig:lower_bounds_2d}
(Left). Again, while GD quickly finds the global minima for $x>x_{c}$,
the S point, for $x<x_{c}$ the global minima $SB$ is more difficult
to approach, and the restriction to the 2d submanifold doesn't seem
to provide a computational advantage, possibly due to the presence
of further spurious minima in this restricted space. 

A further approach is to just evaluate the objective function in eq.~(\ref{eq:aLB2})
on the points of the increasingly refined 2d-grid, without any GD
refinement, and take the lowest of the values obtained. With this
approach, we evaluated grids of up to $m=5000^{2}$ points. Results
are presented in fig. \ref{fig:lower_bounds_2d} (Right).

All of the 3 approaches are in good agreement with each other and
with theoretical predictions.

\begin{figure}[h]
\begin{centering}
\includegraphics[width=0.49\columnwidth]{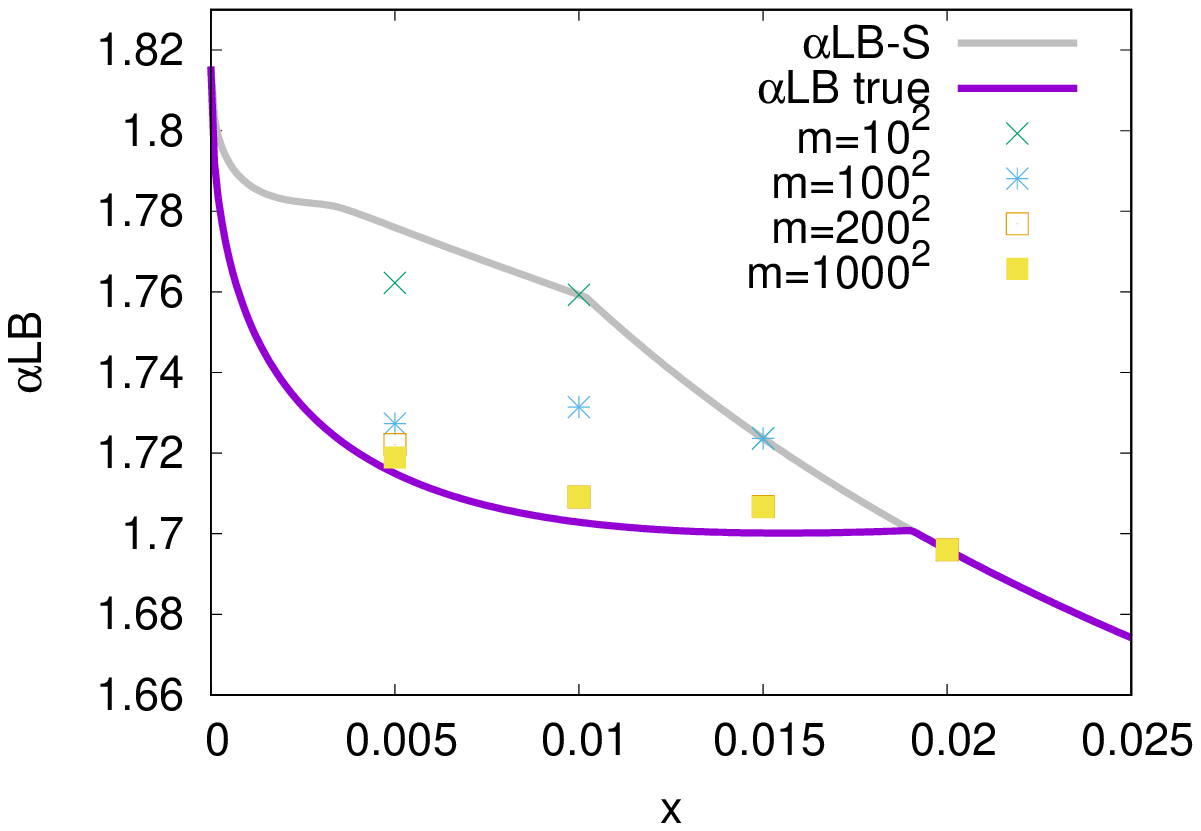}\includegraphics[width=0.49\columnwidth]{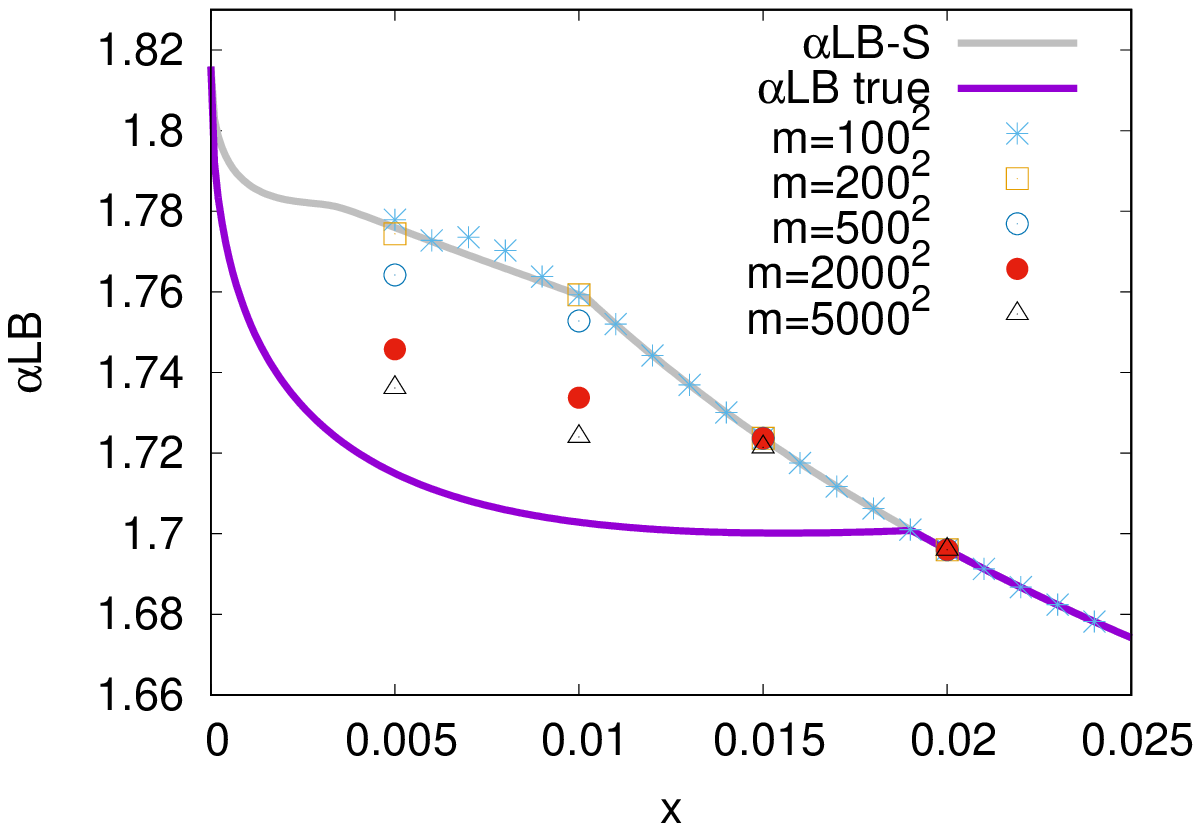}
\par\end{centering}
\caption{\label{fig:lower_bounds_2d}(Left) Numerical and theoretical estimates
for $\alpha_{LB,y=2}\left(x,K=1\right)$ as in fig.~(\ref{fig:Upper-and-lower4d})
but with GD in 2-dimensional space and multiple restarts from grids
of $m$ points. (Right) Evaluation of the points in 2d grids of different
sizes $m$ with no GD refinement.}
\end{figure}

\subsection{Computation of $f_{2}\left(\vec{q}_{0},x,K\right)$}

The computation in an efficient and precise way of the quantity $f_{2}\left(\vec{q}_{0},x,K\right)$
is crucial for the numerical results. We use the Cholesky decomposition
of matrix $\Sigma=C_{L}C_{L}^{T}$ where $C_{L}$ is lower triangular
and $C_{L}^{-1}=C_{L}^{T}$. Then it is natural to use the change
of variable $\mathbf{y=}C_{L}^{-1}\mathbf{z}$, in matrix form 
\[
\left(\begin{array}{c}
z_{1}\\
z_{2}\\
\tilde{z}_{1}\\
\tilde{z}_{2}
\end{array}\right)=\left(\begin{array}{cccc}
c_{11} & 0 & 0 & 0\\
c_{21} & c_{22} & 0 & 0\\
c_{31} & c_{32} & c_{33} & 0\\
c_{41} & c_{42} & c_{43} & c_{44}
\end{array}\right)\left(\begin{array}{c}
y_{1}\\
y_{2}\\
\tilde{y}_{1}\\
\tilde{y}_{2}
\end{array}\right)
\]

the integral is transformed in the following way:

\begin{align*}
 & f_{2}\left(\vec{q}_{0},x,K\right)=\int_{I_{K}^{4}}\frac{dz_{1}dz_{2}d\tilde{z}_{1}d\tilde{z}_{2}}{\left(2\pi\right)^{2}\left|\Sigma\right|^{1/2}}e^{-\frac{1}{2}\mathbf{z}^{T}\Sigma^{-1}\text{\ensuremath{\mathbf{z}}}}\\
 & =\frac{1}{\left(2\pi\right)^{2}}\int_{-\frac{K}{c_{11}}}^{\frac{K}{c_{11}}}dy_{1}\int_{\frac{\left(-K-c_{21}y_{1}\right)}{c_{22}}}^{\frac{\left(K-c_{21}y_{1}\right)}{c_{22}}}dy_{2}\int_{\frac{\left(-K-c_{31}y_{1}-c_{32}y_{2}\right)}{c_{33}}}^{\frac{\left(K-c_{31}y_{1}-c_{32}y_{2}\right)}{c_{33}}}d\tilde{y}_{1}\int_{\frac{\left(-K-c_{41}y_{1}-c_{42}y_{2}-c_{43}\tilde{y}_{1}\right)}{c_{44}}}^{\frac{\left(K-c_{41}y_{1}-c_{42}y_{2}-c_{43}\tilde{y}_{1}\right)}{c_{44}}}d\tilde{y}_{2}e^{-\frac{\textbf{y}^{T}\textbf{y}}{2}}\\
 & =\frac{1}{\left(2\pi\right)^{\frac{3}{2}}}\int_{-\frac{K}{c_{11}}}^{\frac{K}{c_{11}}}dy_{1}\int_{\frac{\left(-K-c_{21}y_{1}\right)}{c_{22}}}^{\frac{\left(K-c_{21}y_{1}\right)}{c_{22}}}dy_{2}\int_{\frac{\left(-K-c_{31}y_{1}-c_{32}y_{2}\right)}{c_{33}}}^{\frac{\left(K-c_{31}y_{1}-c_{32}y_{2}\right)}{c_{33}}}d\tilde{y}_{1}e^{-\frac{y_{1}^{2}+y_{2}^{2}+\tilde{y}_{1}^{2}}{2}}\sum_{s=\pm1}s\,H\left(\frac{-sK-c_{41}y_{1}-c_{42}y_{2}-c_{43}\tilde{y}_{1}}{c_{44}}\right)
\end{align*}
where in the last line we have performed the integral over $\tilde{y}_{2}$,
using the definition $H\left(x\right)=\frac{1}{2}\mathrm{erfc}\left(\frac{x}{\sqrt{2}}\right)$.

\section{$n$-th moment of $y$-solutions multiplet using Replica Ansatz\label{sec:n-th-moment-using}}

Let us define $\mathcal{Z}_{y}$ to be the number of configurations
of $y$ vectors of binary weights each satisfying the CSP eq.~(\ref{eq:CSP})
and whose mutual distance is $x$. In the following we will use the
overlap $q_{1}=1-2x$ between solutions as an external control parameter.
We also introduce for convenience of notation the indicator functions
$\varphi_{K}\left(z\right)=\ind\left(\left|z\right|\le K\right)$
and $\delta\left(z\right)=\ind\left(z=0\right)$. We denote with $\delta_{D}$
the Dirac's delta distribution. With these definitions we have:

\begin{eqnarray*}
\mathcal{Z}_{y}\left(q_{1},K,\boldsymbol{\xi}\right) & = & \sum_{\left\{ \mathbf{w}^{a}\right\} _{a=1}^{y}}\prod_{a=1}^{y}\mathbb{X}_{\xi,K}\left(\textbf{w}^{a}\right)\prod_{a<b}^{y}\delta\left(\sum_{i}w_{i}^{a}w_{i}^{b}-\left\lfloor Nq_{1}\right\rfloor \right)\\
 & = & \sum_{\left\{ \mathbf{w}^{a}\right\} _{a=1}^{y}}\prod_{a=1}^{y}\prod_{\mu=1}^{M}\varphi_{K}\left(\sum_{i}w_{i}^{a}\xi_{i}^{\mu}\right)\prod_{a<b}^{y}\delta\left(\sum_{i}w_{i}^{a}w_{i}^{b}-\left\lfloor Nq_{1}\right\rfloor \right).
\end{eqnarray*}

We want to take the expectation of the $n$-th moment of this partition
function:

\begin{eqnarray*}
\mathcal{Z}_{y}^{n}\left(q_{1},K,\boldsymbol{\xi}\right) & = & \sum_{\left\{ \mathbf{w}_{\alpha}^{a}\right\} }\prod_{a,\alpha,\mu}\varphi_{K}\left(\sum_{i}w_{\alpha,i}^{a}\xi_{i}^{\mu}\right)\prod_{\alpha,a<b}\delta\left(\sum_{i}w_{\alpha,i}^{a}w_{\alpha,i}^{b}-\left\lfloor Nq_{1}\right\rfloor \right)\\
 & = & \sum_{\left\{ \mathbf{w}_{\alpha}^{a}\right\} }\int\prod_{a,\alpha,\mu}d\lambda_{\alpha,\mu}^{a}\varphi_{K}\left(\lambda_{\alpha,\mu}^{a}\right)\delta_{D}\left(\lambda_{\alpha,\mu}^{a}-\sum_{i}w_{\alpha,i}^{a}\xi_{i}^{\mu}\right)\prod_{\alpha,a<b}\delta\left(\sum_{i}w_{\alpha,i}^{a}w_{\alpha,i}^{b}-\left\lfloor Nq_{1}\right\rfloor \right)\\
 & = & \sum_{\left\{ \mathbf{w}_{\alpha}^{a}\right\} }\int\prod_{a,\alpha,\mu}\frac{d\lambda_{\alpha,\mu}^{a}d\hat{\lambda}_{\alpha,\mu}^{a}}{2\pi}\varphi_{K}\left(\lambda_{\alpha,\mu}^{a}\right)e^{i\hat{\lambda}_{\alpha,\mu}^{a}\lambda_{\alpha,\mu}^{a}-i\hat{\lambda}_{\alpha,\mu}^{a}\sum_{i}w_{\alpha,i}^{a}\xi_{i}^{\mu}}\prod_{\alpha,a<b}\delta\left(\sum_{i}w_{\alpha,i}^{a}w_{\alpha,i}^{b}-\left\lfloor Nq_{1}\right\rfloor \right).
\end{eqnarray*}

Now we can take the average over the quenched disorder (in the large
$N$ limit, up to the leading exponential order):

\begin{eqnarray*}
\mathbb{E}\left[\mathcal{Z}_{y}^{n}\left(q_{1},K,\boldsymbol{\xi}\right)\right] & = & \sum_{\left\{ \mathbf{w}_{\alpha}^{a}\right\} }\int\prod_{a,\alpha,\mu}\left(\frac{d\lambda_{\alpha,\mu}^{a}d\hat{\lambda}_{\alpha,\mu}^{a}}{2\pi}\varphi_{K}\left(\lambda_{\alpha,\mu}^{a}\right)e^{i\hat{\lambda}_{\alpha,\mu}^{a}\lambda_{\alpha,\mu}^{a}}\right)\mathbb{E}\left[e^{\sum_{\mu,i}\xi_{i}^{\mu}\sum_{a,\alpha}\left(-i\hat{\lambda}_{\alpha,\mu}^{a}w_{\alpha,i}^{a}\right)}\right]\prod_{\alpha,a<b}\delta\left(\sum_{i}w_{\alpha,i}^{a}w_{\alpha,i}^{b}-\left\lfloor Nq_{1}\right\rfloor \right)\\
 & \cong & \sum_{\left\{ \mathbf{w}_{\alpha}^{a}\right\} }\int\prod_{a,\alpha,\mu}\left(\frac{d\lambda_{\alpha,\mu}^{a}d\hat{\lambda}_{\alpha,\mu}^{a}}{2\pi}\varphi_{K}\left(\lambda_{\alpha,\mu}^{a}\right)e^{i\hat{\lambda}_{\alpha,\mu}^{a}\lambda_{\alpha,\mu}^{a}}\right)e^{-\frac{1}{2N}\sum_{\mu,i}\left(\sum_{a,\alpha}\hat{\lambda}_{\alpha,\mu}^{a}w_{\alpha,i}^{a}\right)^{2}}\prod_{\alpha,a<b}\delta\left(\sum_{i}w_{\alpha,i}^{a}w_{\alpha,i}^{b}-\left\lfloor Nq_{1}\right\rfloor \right)\\
 & = & \sum_{\left\{ \mathbf{w}_{\alpha}^{a}\right\} }\int\prod_{a,\alpha,\mu}\left(\frac{d\lambda_{\alpha,\mu}^{a}d\hat{\lambda}_{\alpha,\mu}^{a}}{2\pi}\varphi_{K}\left(\lambda_{\alpha,\mu}^{a}\right)\right)e^{i\sum_{\alpha,a,\mu}\hat{\lambda}_{\alpha,\mu}^{a}\lambda_{\alpha,\mu}^{a}-\frac{1}{2}\sum_{\mu}\sum_{a,b}\sum_{\alpha,\beta}\hat{\lambda}_{\alpha,\mu}^{a}\hat{\lambda}_{\beta,\mu}^{b}\left(\frac{\sum_{i}w_{\alpha,i}^{a}w_{\beta,i}^{b}}{N}\right)}\times\\
 &  & \times\prod_{\alpha,a<b}\delta\left(\sum_{i}w_{\alpha,i}^{a}w_{\alpha,i}^{b}-\left\lfloor Nq_{1}\right\rfloor \right).
\end{eqnarray*}

Next, we introduce the overlaps $q_{\alpha\beta}^{ab}=\frac{\sum_{i}w_{\alpha,i}^{a}w_{\beta,i}^{b}}{N}$
via Dirac deltas. :

\begin{eqnarray*}
 & = & \sum_{\left\{ \mathbf{w}_{\alpha}^{a}\right\} }\int\prod_{a,\alpha,\mu}\left(\frac{d\lambda_{\alpha,\mu}^{a}d\hat{\lambda}_{\alpha,\mu}^{a}}{2\pi}\varphi_{K}\left(\lambda_{\alpha,\mu}^{a}\right)\right)\int\prod_{\alpha<\beta;a,b}dq_{\alpha\beta}^{ab}\int\prod_{\alpha;a<b}dq_{\alpha\alpha}^{ab}e^{i\sum_{\alpha,a,\mu}\hat{\lambda}_{\alpha,\mu}^{a}\lambda_{\alpha,\mu}^{a}-\sum_{\mu}\sum_{a,b,\alpha<\beta}\hat{\lambda}_{\alpha,\mu}^{a}\hat{\lambda}_{\beta,\mu}^{b}q_{\alpha\beta}^{ab}}\times\\
 &  & e^{-\sum_{\mu}\sum_{\alpha,a<b}\hat{\lambda}_{\alpha,\mu}^{a}\hat{\lambda}_{\beta,\mu}^{b}q_{1}-\frac{1}{2}\sum_{\mu}\sum_{a,\alpha}\left(\hat{\lambda}_{\alpha,\mu}^{a}\right)^{2}}\prod_{\alpha<\beta;a,b}\delta_{D}\left(\frac{\sum_{i}w_{\alpha,i}^{a}w_{\beta,i}^{b}}{N}-q_{\alpha\beta}^{ab}\right)\prod_{\alpha,a<b}\delta_{D}\left(\frac{\sum_{i}w_{\alpha,i}^{a}w_{\alpha,i}^{b}}{N}-q_{\alpha\alpha}^{ab}\right)\delta\left(Nq_{\alpha\alpha}^{ab}-\left\lfloor Nq_{1}\right\rfloor \right)\\
 & \cong & \sum_{\left\{ \mathbf{w}_{\alpha}^{a}\right\} }\int\prod_{\alpha<\beta;a,b}\frac{dq_{\alpha\beta}^{ab}d\hat{q}_{\alpha\beta}^{ab}}{2\pi}\int\prod_{\alpha;a<b}\frac{d\hat{q}_{\alpha\alpha}^{ab}}{2\pi}\int\prod_{a,\alpha,\mu}\left(\frac{d\lambda_{\alpha,\mu}^{a}d\hat{\lambda}_{\alpha,\mu}^{a}}{2\pi}\varphi_{K}\left(\lambda_{\alpha,\mu}^{a}\right)\right)\times\\
 &  & \times e^{i\sum_{\alpha,a,\mu}\hat{\lambda}_{\alpha,\mu}^{a}\lambda_{\alpha,\mu}^{a}-\sum_{\mu}\sum_{a,b,\alpha<\beta}\hat{\lambda}_{\alpha,\mu}^{a}\hat{\lambda}_{\beta,\mu}^{b}q_{\alpha\beta}^{ab}-\sum_{\mu}\sum_{\alpha,a<b}\hat{\lambda}_{\alpha,\mu}^{a}\hat{\lambda}_{\beta,\mu}^{b}q_{1}-\frac{1}{2}\sum_{\mu}\sum_{a,\alpha}\left(\hat{\lambda}_{\alpha,\mu}^{a}\right)^{2}-N\sum_{\alpha<\beta;a,b}\hat{q}_{\alpha\beta}^{ab}q_{\alpha\beta}^{ab}}\times\\
 &  & \times e^{\sum_{\alpha<\beta;a,b}\hat{q}_{\alpha\beta}^{ab}\sum_{i}w_{\alpha,i}^{a}w_{\beta,i}^{b}-Nq_{1}\sum_{\alpha,a<b}\hat{q}_{\alpha}^{ab}+\sum_{\alpha,a<b}\hat{q}_{\alpha\alpha}^{ab}\sum_{i}w_{\alpha,i}^{a}w_{\alpha,i}^{b}}\\
 & = & \int\prod_{\alpha<\beta;a,b}\frac{dq_{\alpha\beta}^{ab}d\hat{q}_{\alpha\beta}^{ab}}{2\pi}\prod_{\alpha;a<b}\frac{d\hat{q}_{\alpha\alpha}^{ab}}{2\pi}e^{N\left(G_{I}\left(q,\hat{q}\right)+G_{S}\left(\hat{q}\right)+\alpha G_{E}\left(q\right)\right)},
\end{eqnarray*}
where we have introduced the interaction, entropic and energetic terms:

\begin{eqnarray*}
G_{I}^{n,y}\left(q,\hat{q}\right) & = & -\sum_{\alpha<\beta;a,b}\hat{q}_{\alpha\beta}^{ab}q_{\alpha\beta}^{ab}-q_{1}\sum_{\alpha,a<b}\hat{q}_{\alpha}^{ab}\\
G_{S}^{n,y}\left(\hat{q}\right) & = & \frac{1}{N}\ln\sum_{\left\{ \mathbf{w}_{\alpha}^{a}\right\} }e^{\sum_{\alpha<\beta;a,b}\hat{q}_{\alpha\beta}^{ab}\sum_{i}w_{\alpha,i}^{a}w_{\beta,i}^{b}+\sum_{\alpha,a<b}\hat{q}_{\alpha\alpha}^{ab}\sum_{i}w_{\alpha,i}^{a}w_{\alpha,i}^{b}}\\
G_{E}^{n,y,K}\left(q\right) & = & \frac{1}{\alpha N}\ln\int\prod_{a,\alpha,\mu}\left(\frac{d\lambda_{\alpha,\mu}^{a}d\hat{\lambda}_{\alpha,\mu}^{a}}{2\pi}\varphi_{K}\left(\lambda_{\alpha,\mu}^{a}\right)\right)e^{i\sum_{\alpha,a,\mu}\hat{\lambda}_{\alpha,\mu}^{a}\lambda_{\alpha,\mu}^{a}-\sum_{\mu}\sum_{a,b,\alpha<\beta}\hat{\lambda}_{\alpha,\mu}^{a}\hat{\lambda}_{\beta,\mu}^{b}q_{\alpha\beta}^{ab}}\times\\
 &  & \times e^{-\sum_{\mu}\sum_{\alpha,a<b}\hat{\lambda}_{\alpha,\mu}^{a}\hat{\lambda}_{\beta,\mu}^{b}q_{1}-\frac{1}{2}\sum_{\mu}\sum_{a,\alpha}\left(\hat{\lambda}_{\alpha,\mu}^{a}\right)^{2}}
\end{eqnarray*}

We introduce a replica-symmetric ansatz on the matrices $Q_{\alpha\beta}$
and $\hat{Q}_{\alpha\beta}$ which is specified by the following set
of equations: 
\[
Q_{\alpha\beta}^{ab}=\begin{cases}
1 & \text{if }\alpha=\beta\text{ and }a=b\\
q_{0} & \text{if }\alpha\neq\beta\\
q_{1} & \text{if }\alpha=\beta\text{ and }a\neq b
\end{cases}\qquad\hat{Q}_{\alpha\beta}^{ab}=\begin{cases}
0 & \text{if }\alpha=\beta\text{ and }a=b\\
\hat{q}_{0} & \text{if }\alpha\neq\beta\\
\hat{q}_{1} & \text{if }\alpha=\beta\text{ and }a\neq b
\end{cases}.
\]
In the case $y=3$ and $n=2$ they look as follows:

\[
Q=\left(\begin{array}{cccccc}
1 & q_{1} & q_{1} & q_{0} & q_{0} & q_{0}\\
q_{1} & 1 & q_{1} & q_{0} & q_{0} & q_{0}\\
q_{1} & q_{1} & 1 & q_{0} & q_{0} & q_{0}\\
q_{0} & q_{0} & q_{0} & 1 & q_{1} & q_{1}\\
q_{0} & q_{0} & q_{0} & q_{1} & 1 & q_{1}\\
q_{0} & q_{0} & q_{0} & q_{1} & q_{1} & 1
\end{array}\right)\qquad\hat{Q}=\left(\begin{array}{cccccc}
0 & \hat{q}_{1} & \hat{q}_{1} & \hat{q}_{0} & \hat{q}_{0} & \hat{q}_{0}\\
\hat{q}_{1} & 0 & \hat{q}_{1} & \hat{q}_{0} & \hat{q}_{0} & \hat{q}_{0}\\
\hat{q}_{1} & \hat{q}_{1} & 0 & \hat{q}_{0} & \hat{q}_{0} & \hat{q}_{0}\\
\hat{q}_{0} & \hat{q}_{0} & \hat{q}_{0} & 0 & \hat{q}_{1} & \hat{q}_{1}\\
\hat{q}_{0} & \hat{q}_{0} & \hat{q}_{0} & \hat{q}_{1} & 0 & \hat{q}_{1}\\
\hat{q}_{0} & \hat{q}_{0} & \hat{q}_{0} & \hat{q}_{1} & \hat{q}_{1} & 0
\end{array}\right).
\]

We now compute the interaction, entropic and energetic terms using
this ansatz: 

\begin{eqnarray}
G_{I}^{n,y}\left(q_{0},q_{1},\hat{q}_{0},\hat{q}_{1}\right) & = & -y^{2}\frac{n\left(n-1\right)}{2}q_{0}\hat{q}_{0}-n\frac{y\left(y-1\right)}{2}q_{1}\hat{q}_{1}-\frac{yn}{2}\hat{q}_{1}\label{eq:interactionG}
\end{eqnarray}

\begin{eqnarray}
G_{S}^{n,y}\left(\hat{q}_{0},\hat{q}_{1}\right) & = & \frac{1}{N}\ln\sum_{\left\{ \mathbf{w}_{\alpha}^{a}\right\} }\prod_{i}e^{\sum_{\alpha<\beta;a,b}\hat{q}_{0}w_{\alpha,i}^{a}w_{\beta,i}^{b}+\sum_{\alpha,a<b}\hat{q}_{1}w_{\alpha,i}^{a}w_{\alpha,i}^{b}}\label{eq:entropicG}\\
 & = & -\frac{ny\hat{q}_{1}}{2}+\ln\sum_{\left\{ w_{\alpha}^{a}\right\} }e^{\frac{1}{2}\hat{q}_{0}\left(\sum_{a\alpha}w_{\alpha}^{a}\right)^{2}+\frac{\hat{q}_{1}-\hat{q}_{0}}{2}\sum_{\alpha}\left(\sum_{a}w_{\alpha}^{a}\right)^{2}}\nonumber \\
 & = & -\frac{ny\hat{q}_{1}}{2}+\ln\sum_{\left\{ w_{\alpha}^{a}\right\} }\int Dz\ e^{z\sqrt{\hat{q}_{0}}\sum_{a\alpha}w_{\alpha}^{a}}\int\prod_{\alpha}Dt_{\alpha}e^{\sqrt{\hat{q}_{1}-\hat{q}_{0}}\sum_{\alpha}t_{\alpha}\sum_{a}w_{\alpha}^{a}}\nonumber \\
 & = & -\frac{ny\hat{q}_{1}}{2}+\ln\int Dz\left[\int Dt\left(2\cosh\left(\sqrt{\hat{q}_{0}}z+\sqrt{\hat{q}_{1}-\hat{q}_{0}}t\right)\right)^{y}\right]^{n}\nonumber 
\end{eqnarray}
\begin{eqnarray}
G_{E}^{n,y,K}\left(q_{0},q_{1}\right) & = & \frac{1}{\alpha N}\ln\int\prod_{a,\alpha,\mu}\left(\frac{d\lambda_{\alpha,\mu}^{a}d\hat{\lambda}_{\alpha,\mu}^{a}}{2\pi}\varphi_{K}\left(\lambda_{\alpha,\mu}^{a}\right)\right)e^{i\sum_{\alpha,a,\mu}\hat{\lambda}_{\alpha,\mu}^{a}\lambda_{\alpha,\mu}^{a}-\sum_{\mu}q_{0}\sum_{a,b,\alpha<\beta}\hat{\lambda}_{\alpha,\mu}^{a}\hat{\lambda}_{\beta,\mu}^{b}}\times\label{eq:energeticG}\\
 &  & \times e^{-\sum_{\mu}\sum_{\alpha,a<b}\hat{\lambda}_{\alpha,\mu}^{a}\hat{\lambda}_{\beta,\mu}^{b}q_{1}-\frac{1}{2}\sum_{\mu}\sum_{a,\alpha}\left(\hat{\lambda}_{\alpha,\mu}^{a}\right)^{2}}\nonumber \\
 & = & \ln\int\prod_{a,\alpha}\left(\frac{d\lambda_{\alpha}^{a}d\hat{\lambda}_{\alpha}^{a}}{2\pi}\varphi_{K}\left(\lambda_{\alpha}^{a}\right)\right)e^{i\sum_{\alpha,a}\hat{\lambda}_{\alpha}^{a}\lambda_{\alpha}^{a}-\frac{1}{2}q_{0}\left(\sum_{a\alpha}\hat{\lambda}_{\alpha}^{a}\right)^{2}-\frac{q_{1}-q_{0}}{2}\sum_{\alpha}\left(\sum_{a}\hat{\lambda}_{\alpha}^{a}\right)^{2}-\frac{1-q_{1}}{2}\sum_{a\alpha}\left(\hat{\lambda}_{\alpha}^{a}\right)^{2}}\nonumber \\
 & = & \ln\int Dz\int\prod_{\alpha}Dt_{\alpha}\int\prod_{a\alpha}\left(\frac{d\lambda_{\alpha}^{a}d\hat{\lambda}_{\alpha}^{a}}{2\pi}\varphi_{K}\left(\lambda_{\alpha}^{a}\right)\right)e^{i\sum_{\alpha,a}\hat{\lambda}_{\alpha}^{a}\lambda_{\alpha}^{a}+iz\sqrt{q_{0}}\sum_{a\alpha}\hat{\lambda}_{\alpha}^{a}+i\sqrt{q_{1}-q_{0}}\sum_{\alpha}t_{\alpha}\sum_{a}\hat{\lambda}_{\alpha}^{a}-\frac{1-q_{1}}{2}\sum_{a\alpha}\left(\hat{\lambda}_{\alpha}^{a}\right)^{2}}\nonumber \\
 & = & \ln\int Dz\left[\int Dt\left[\int\frac{d\lambda d\hat{\lambda}}{2\pi}\varphi_{K}\left(\lambda\right)e^{i\hat{\lambda}\lambda+iz\sqrt{q_{0}}\hat{\lambda}+i\sqrt{q_{1}-q_{0}}t\hat{\lambda}-\frac{1-q_{1}}{2}\hat{\lambda}^{2}}\right]^{y}\right]^{n}\nonumber \\
 & = & \ln\int Dz\left[\int Dt\left[\int\frac{d\lambda}{\sqrt{2\pi\left(1-q_{1}\right)}}\varphi_{K}\left(\lambda\right)e^{-\frac{\left(\lambda+\sqrt{q_{0}}z+\sqrt{q_{1}-q_{0}}t\right)^{2}}{2\left(1-q_{1}\right)}}\right]^{y}\right]^{n}\nonumber \\
 & = & \ln\int Dz\left[\int Dt\left[\sum_{s=\pm1}s\,H\left(\frac{-s\,K}{\sqrt{1-q_{1}}}+\frac{\sqrt{q_{0}}z+\sqrt{q_{1}-q_{0}}t}{\sqrt{1-q_{1}}}\right)\right]^{y}\right]^{n}\nonumber 
\end{eqnarray}

In the last line, as in the main text, the function $H\left(x\right)$
is defined as $H\left(x\right)\equiv\int_{x}^{\infty}Dz\equiv\int_{x}^{\infty}\,\frac{dz}{\sqrt{2\pi}}\,e^{-z^{2}/2}=\frac{1}{2}\mathrm{erfc}\left(\frac{x}{\sqrt{2}}\right).$
\end{document}